\documentclass[a4paper,11pt]{article}

\usepackage{amsthm,amsmath,amssymb,amsfonts,pstricks,pst-node,graphics,graphicx}

\def\im{\operatorname{Im}}

\newcommand\Der{{\rm Der}}
\newcommand\Ord{{\rm Ord}}
\newcommand\cS{{\mathcal S}}
\newcommand\cE{{\mathcal E}}

\newcommand\la{{\lambda}}

\renewcommand\d{{\rm d}}

\newcommand\cM{{\mathcal M}}

\newcommand\cK{{\mathrm k}}

\newcommand\cF{{\mathrm F}}
\newcommand\cG{{\mathcal G}}
\newcommand\cC{{\mathcal C}}

\newcommand\cR{{\mathcal R}}

\newcommand\cQ{{\mathcal Q}}
\newcommand\cA{{\mathcal A}}
\newcommand\cB{{\mathcal B}}
\newcommand\cI{{\mathcal I}}

\newcommand\cD{{\mathcal D}}
\newcommand\cL{{\mathcal L}}
\newcommand\cT{{\mathcal T}}

\def\bbbc{{\mathbb C}}
\def\bbbz{{\mathbb Z}}
\def\bbbr{{\mathbb R}}
\def\bbbk{{\mathbb K}}
\def\bbbd{{\mathbb D}}

\def\f{{\bf f}}
\def\bu{{\bf u}}
\def\bU{{\bf U}}
\def\bV{{\bf V}}
\def\bW{{\bf W}}
\def\g{{\bf g}}
\def\h{{\bf h}}
\def\a{{\bf a}}

\def\Id{\mbox{Id}}

\newtheorem{Rem}{Remark}

\newtheorem{Def}{Definition}
\newtheorem{The}{Theorem}
\newtheorem{Pro}{Proposition}
\newtheorem{Lem}{Lemma}
\newtheorem{Ex}{Example}
\newtheorem{Cor}{Corollary}

\begin{document}

\title{Rational recursion operators for integrable differential-difference 
equations}
\author{Sylvain Carpentier $^\ddagger$, Alexander V. Mikhailov$^{\star}$ and 
Jing Ping 
Wang$ ^\dagger $
\\
$\ddagger$ Mathematics Department, Columbia University, USA
\\
$\dagger$ School of Mathematics, Statistics \& Actuarial Science, University of 
Kent, UK \\
$\star$ Applied Mathematics Department, University of Leeds, UK
}
\date{}  

\maketitle
\begin{abstract}
In this paper we introduce preHamiltonian pairs of 
difference operators  and study  their connections with Nijenhuis 
operators and the existence of weakly non-local inverse recursion operators for 
differential--difference equations. We begin with a rigorous setup of the 
problem in terms of the skew field $\cQ$ of {\em rational} (pseudo--difference)
operators over a difference field $\cF$ with a zero characteristic subfield of 
constants $\cK\subset\cF$ and the principal ideal ring $\cM_n(\cQ)$ of matrix 
rational (pseudo--difference) operators. In particular, we give a criteria for 
a 
rational operator to be weakly non--local. A difference operator $H$ is called 
preHamiltonian, if its image is a Lie $\cK$--subalgebra with respect the the 
Lie bracket on $\cF$. Two preHamiltonian operators form a preHamiltonian pair 
if 
any $\cK$--linear combination of them is preHamiltonian. Then we show that a preHamiltonian pair naturally leads to a
Nijenhuis operator, and a Nijenhuis operator can be represented in terms of a preHamiltonian pair. 
This provides a systematical method to check  whether a rational operator is Nijenhuis. As an application, 
we construct a preHamiltonian pair and thus a Nijenhuis recursion operator for 
the  differential-difference
equation recently discovered by Adler \& Postnikov. The Nijenhuis 
operator  obtained is not weakly non-local. We prove that it  generates 
an 
infinite hierarchy of local commuting symmetries.  We also illustrate our 
theory on the well known examples including the Toda, the Ablowitz-Ladik and 
the Kaup-Newell differential-difference equations.
\end{abstract}

\section{Introduction}

The existence of an 
infinite hierarchy of commuting symmetries is one of a characteristic property 
of integrable  systems. Symmetries can be 
generated
by recursion operators \cite{AKNS74,mr58:25341}, which  are often 
pseudo--differential and map a symmetry to a new
symmetry. An important property of recursion operators, called the Nijenhuis property, is to generate an 
abelian Lie algebra of symmetries. Such property has been independently studied by Fuchssteiner 
\cite{Fuc79} and Magri \cite{Mag80}. 
To prove that  a pseudo--differential operator is a Nijenhuis operator and it generates an 
infinite hierarchy of
local symmetries is a challenging problem. In the most common  case of weakly 
non-local Nijenhuis operators this problem has been addressed in 
\cite{mr1974732, serg5, wang09}. The relations between  bi-Hamiltonian structures and 
Nijenhuis operators have been studied in papers of 
Gel'fand \& Dorfman \cite{GD79, mr94j:58081} and Fuchssteiner \& Fokas
\cite{mr82g:58039,mr84j:58046}. Recently a rigorous approach to  
pseudo--differential Hamiltonian % and Nijenhuis 
operators have been developed in 
the series of papers by Kac and his co-authors \cite{BDSK09, DSK13, DSKV16}.

The theory of integrable differential-difference equations is much less 
developed. The basic concepts such as symmetries, 
conservation laws and  Hamiltonian operators were formulated in the frame  of a
variational complex in \cite{kp85}. 
The aim of this paper is to build up a rigorous setting for rational matrix 
(pseudo--difference) operators suitable for the study of integrable 
differential-difference systems. We introduce and study preHamiltonian pairs of 
difference operators, their connections with  Nijenhuis 
operators and the existence of weakly non-local inverse recursion operators for 
differential--difference equations.

Let us consider the well-known Volterra chain
\begin{eqnarray}\label{vol}
u_t=u (u_1-u_{-1}) ,
\end{eqnarray}
where $u$ is a function of a lattice variable $n\in \bbbz$ and time $t$. Here
we use the notations
\begin{eqnarray*}
 u_t=\partial_t(u), \quad u_j=\cS^j u(n,t)=u(n+j,t)
\end{eqnarray*}
and $\cS$ is the shift operator. It possesses a recursion operator
$$R=u \cS +u +u_1+u \cS^{-1}+u(u_1-u_{-1}) (\cS-1)^{-1}\frac{1}{u} ,
$$
where \((\cS-1)^{-1} \) stands for the  inverse of \(\cS-1\). Thus this 
recursion
operator is only defined on \(u\im (\cS-1)\). It is a Nijenhuis operator and 
generates a commutative hierarchy of
symmetries:
\[
u_{t_j}=R^j (u_t)=R^j\left(u(u_1-u_{-1})\right),\qquad j=0, 1, 2,\cdots .
\]
The concept of Hamiltonian pairs was introduced by Magri \cite{Mag78}. He found 
that
some systems admitted two distinct but compatible Hamiltonian structures 
(a Hamiltonian
pair) and named them twofold Hamiltonian system, nowadays known as 
bi-Hamiltonian
systems. The Volterra chain is a bi-Hamiltonian system and it can be written
\[
u_t=H_1\ \delta_u u=H_2\ \delta_u \frac{\ln u}{2},
\]
where $\delta_u$ is variational derivative with respect to the dependent 
variable $u$
and two difference operators
\begin{eqnarray*}
&&H_1=u(\cS-\cS^{-1})u;\\
&&H_2=u(\cS u\cS+u \cS+\cS u-u\cS^{-1}-\cS^{-1} u-\cS^{-1} u\cS^{-1})u
\end{eqnarray*}
form a Hamiltonian pair.
The Nijenhuis recursion operator of 
the
Volterra chain can be obtained via the Hamiltonian pair, that is, $
R=H_2 H_1^{-1}$. This decomposition is known as the Lenard scheme used to
construct the hierarchies of infinitely many symmetries and cosymmetries.

Notice that the above difference operators have a right common factor:
\begin{eqnarray*}
&&H_1=u (\cS-1)(1+\cS^{-1})u;\qquad
H_2
=u(\cS u\cS+u \cS- u-\cS^{-1} u) (1+\cS^{-1}) u .
\end{eqnarray*}
This implies that 
\begin{equation}\label{volab}
R=A B^{-1}, \  \mbox{where $A=u(\cS +1)(u \cS-\cS^{-1} u)$ and $B=u 
(\cS-1)$}. 
\end{equation}
Here operators $A$ and $B$ are not anti-symmetric, and thus not Hamiltonian. 
However, like in the case of  
Hamiltonian pairs, the image of $A$ and $B$, as well as the image of 
linear combinations of these two operators,
form a Lie subalgebra. Such operators we call preHamiltonian. In this 
paper, we will explore properties of preHamiltonian operators and their
relations with Nijenhuis operators. For the differential case 
some of these results have been obtained in
\cite{Carp2017}.
The main difference between differential operators and difference operators lies 
in that the total derivative is a 
derivation and the shift operator $\cS$ is an automorphism. The set of 
invertible difference operators is much richer than in the differential case. 
In the scalar case all difference operators of 
the form  $a \cS^j$, where $a$ is a difference function and $j\in \mathbb{Z}$,
are invertible, while in the differential case, the only invertible operators 
are operators of multiplication by a function. The definition of the order of  
difference and differential operators are essentially different.

The arrangement of this paper is as follows: In Section \ref{Sec2}, we define a 
difference field $\cF$, the Lie algebra $\cA$ of its evolutionary 
derivations (or evolutionary vector fields) which is a subalgebra of ${\rm 
Der}\,\cF$ and discuss algebraic properties of the noncommutative ring of 
difference operators. In 
particular, we show that it is a right and left 
Euclidean domain and satisfies the right (left) Ore property. Then we define the 
skew field of rational (pseudo--difference) operators, i.e. operators  
of the form $AB^{-1}$, where $A$ and $B$ are difference operators.
Next we discuss the relation between rational operators and weakly 
nonlocal operators, namely we formulate a criteria for a rational operator to 
be weakly nonlocal.
Finally we adapt all these results to rational matrix difference operators by 
defining the order of the operator as the order of its Dieudonn\'e determinant. In 
Section \ref{Sec3} we define  preHamiltonian  difference operators as operators 
on $\cF$ whose images define a Lie subalgebra in $\cA$. We  explore the 
interrelation between preHamiltonian pairs and Nijenhuis 
operators. We show that 
if operators $A$ and $B$ form a preHamiltonian pair,
then $R=AB^{-1}$ is Nijenhuis. Conversely, if $R$ is Nijenhuis and $B$ is 
preHamiltonian, then $A$ and $B$ form a 
preHamiltonian pair.  
These two sections are the theoretical foundation of the paper. In Section 
\ref{Sec4}, we give basic definitions such as
 symmetries, recursion operators and Hamiltonian for differential-difference 
equations. We also show how operators $A$ 
and $B$ are related to the equation if $AB^{-1}$ is its recursion operator. Next 
two sections are applications of the theoretical results in Section \ref{Sec2} and \ref{Sec3} for integrable 
differential- 
difference equations. In Section 
\ref{Sec5}, we construct a recursion operator for a new integrable equation 
derived by Adler and Postilion in 
\cite{adler2}:
$$
u_t=u^2(u_2u_1-u_{-1}u_{-2})-u(u_1-u_{-1}),
$$
using its Lax representation presented in the same paper. The obtained recursion 
operator is no longer weakly nonlocal. 
We show that it is indeed Nijenhuis by rewriting it as a rational difference 
operator and that it generates 
infinitely many commuting local symmetries.  To improve the readability,  we put some technical lemmas 
used for the proof of the main result on the locality of commuting symmetries
in appendix B.
For some integrable 
differential-difference equations, such as the Ablowitz-Ladik Lattice 
\cite{ablowitz2}, the recursion operator and its 
inverse are both weakly nonlocal.
In Section \ref{Sec6}, we apply the 
theoretical results from Section \ref{Sec2} 
to check whether the inverse recursion operators are weakly nonlocal, and if 
so, we demonstrate how to cast them in the weakly nonlocal form. To illustrate 
the method we choose 
four typical examples. However, the method is general and it can be applied to 
any integrable  differential--difference system, including all systems listed 
in 
\cite{kmw13}. At the end of the paper we give a short conclusion and discussion 
on  our new results on relation between preHamiltonian and Hamiltonian 
operators. 
To be self-contained, we also include Appendix A, containing some basic definitions 
for a unital non-commutative ring.

\section{Algebraic properties of difference operators}\label{Sec2}
In this section, we give a definition of rational difference operators and 
explore their properties.
The main objects of our study in this paper are systems of 
evolutionary differential-difference equations and hidden structures associated with 
them. We first consider the scalar case. A generalization to the 
multi-component case will be discussed in the 
end of this section.

\subsection{Difference field and evolutionary vector fields}

Let $\cK $ be a zero characteristic base field,  such as 
$\bbbc$  or $\bbbr$. We define the polynomial ring 
$$\mathrm{K}=\cK[\ldots, u _{-1}, u _0, u _1,\ldots]$$  of the infinite set of 
variables 
$\{u\}=\{ u _k;\ k\in\bbbz\}$ and the corresponding field of fractions 
$$\cF=\cK(\ldots, u _{-1}, u _0, u _1,\ldots).$$ It is assumed that every element 
of  
$\mathrm{K}$ and $\cF$ depends on a finite number of variables only. 
%In the notation $a=a(u_{i_1},u_{i_2},\ldots,u_{i_k})$ we shall assume that 
%$i_1<i_2<\cdots<i_k$. 

There is a natural automorphism $\cS$ of the field $\cF$, 
which we call the shift operator, defined as 
\[
 \cS: a( u _k,\ldots , u _r)\mapsto a( u _{k+1},\ldots , u _{r+1}),\quad 
\cS:\alpha\mapsto\alpha, \qquad a( u _k,\ldots , u _r)\in \cF,\ \ \alpha\in\cK.
\]
For $a=a( u _k,\ldots , u _r)\in\cF$ we will often use notation 
\[
 a_i=\cS^i(a)=a( u _{k+i},\ldots , u _{r+i}),\qquad 
i\in\bbbz ,
\]
and omit index zero at $a_0$ or $u_0$ when there is no ambiguity.
The field $\cF$  equipped with   the automorphism $\cS$ is a 
difference field  and 
the base field $\cK$ is its subfield  of constants. 

The reflection 
$\cT$ of the lattice $\bbbz$ defined by
\[
 \cT: a( u _k,\ldots , u _r)\mapsto a(u_{-k},\ldots ,u_{-r}),\quad 
\cT:\alpha\mapsto\alpha, \qquad a( u _k,\ldots , u _r)\in \cF,\ \ \alpha\in\cK,
\]
is another 
obvious automorphism of $\cF$ and $\mathrm{K}$. The 
composition $\cS\cT\cS\cT=\Id$ is the identity map. Thus the automorphisms
$\cS,\cT$ generate the infinite dihedral group $\bbbd_\infty$ and the subgroup 
generated by $\cS$ is normal.

The automorphism $\cT$ defines a $\bbbz_2$ grading of the difference field $\cF$ (and 
ring $\mathrm{K}\subset\cF$):
\[
 \cF=\cF_0\oplus\cF_1,\qquad \cF_0\cdot\cF_0=\cF_0,\ \ 
\cF_0\cdot\cF_1=\cF_1,\ \ \cF_1\cdot\cF_1=\cF_0, 
 \]
where $ \cF_k=\{a\in\cF\,|\, \cT(a)=(-1)^k a\}$.

Partial derivatives $\frac{\partial}{\partial u_i},\ i\in\bbbz$ are 
commuting derivations of $\cF$ satisfying the conditions
\begin{equation}\label{spart}
\cS \frac{\partial}{\partial u_{i}}= \frac{\partial}{\partial u_{i+1}} \cS, 
\qquad \cT \frac{\partial}{\partial u_{i}}= \frac{\partial}{\partial u_{-i}} 
\cT.
\end{equation}

A derivation of $\cF$ is said to be evolutionary if it commutes with
the shift operator $\cS$. Such derivation is completely determined 
by one element of $f\in \cF$ and is of the form
\begin{equation}
 \label{Xf}
 X_f=\sum_{i\in\bbbz}\cS^i(f) \frac{\partial}{\partial u_{i}},\qquad f\in\cF.
\end{equation}
An element $f$ is called the characteristic of the evolutionary derivation $X_f$. 
The action of $X_f(a)$ for  $a\in\cF$ can also be represented in the form
\[
 X_f(a)=a_*[f],
\]
where $a_*[f]$ is the Fr\'echet derivative of $a=a(u_p,\ldots,u_q)$ in the 
direction $f$, which is defined as 
\[
 a_*[f]:= \frac{d}{d\epsilon}a(u_p+\epsilon f_p,\ldots,u_q+\epsilon 
f_q)\arrowvert_{\epsilon=0}=\sum_{i=p}^q\frac{\partial a}{\partial u_i}f_i.
\]
The Fr\'echet derivative of $a=a(u_p,\ldots,u_q)$ is
a difference operator represented by a 
finite sum
\begin{equation}\label{astar}
a_*=\sum_{i=p}^q\frac{\partial a}{\partial u_i}\cS^i.
\end{equation}
It is obvious that
\[
 (\cT a)_*=\sum_{i=p}^q\cT\left(\frac{\partial a}{\partial u_i}\right)\cS^{-i}.
\]

Evolutionary derivations form a Lie subalgebra $\cA$ in the the Lie algebra 
$\Der \,
\cF$. Indeed, 
\[\begin{array}{l}
   \alpha X_f+\beta X_g=X_{\alpha f+\beta g},\qquad \alpha,\beta\in\cK,\\
\phantom{}   [X_f,X_g]=X_{[f,g]},
  \end{array}
\]
where $[f,g]\in\cF$ denotes the Lie bracket
\begin{equation}\label{bracket}
 [f,g]=X_f(g)-X_g(f)=g_*[f]-f_*[g]. 
\end{equation}
Lie bracket (\ref{bracket}) is $\cK$--bilinear, anti-symmetric and satisfies 
the Jacobi identity. Thus $\cF$, equipped with the bracket (\ref{bracket}), has 
a structure of  a Lie algebra over $\cK$.

The reflection $\cT$ acts naturally on evolutionary vector derivations
\[
 \cT:\, X_f\mapsto X_{\cT(f)}=\cT\cdot X_f\cdot \cT\, .
 \]
Thus the $\cA$ is a graded Lie algebra
\[
 \cA=\cA_0\oplus\cA_1,\qquad [\cA_0,\cA_0]\subset \cA_0,\ [\cA_0,\cA_1]\subset 
\cA_1,\ [\cA_1,\cA_1]\subset \cA_0,\ 
\]
where $\cA_k=\{X\in\cA\,|\, \cT(X)=(-1)^k X\}$.

\subsection{Rational Difference Operators}\label{rdo}

In this section we give  definitions of difference operators and 
rational pseudo--difference operators, which for simplicity we shall call 
rational operators. 

\begin{Def}\label{deford} (1) A difference operator $B$ of order ${\rm ord}\,  
B:=(M,N)$ with 
coefficients in
$\cF$ is a finite sum of the form
\begin{equation}\label{operB}
B= b_N \cS^{N}+b_{N-1} \cS^{N-1}+\cdots +b_M \cS^{M},\qquad b_k\in\cF, \ \ M\le 
N,\
\ N,M\in\mathbb{Z}.
\end{equation}
The total order of $B$ is defined as ${\rm Ord}B=N-M$. The total order of 
the zero 
operator is defined as $\Ord\, 0:=\{\infty\}$.
\end{Def}

The Fr\'echet derivative (\ref{astar}) is an example of a difference operator 
of order $(p,q)$ and total order ${\rm Ord}\, a_* =q-p$. For an element  
$f\in \cF$ the order and total 
order are defined as ${\rm ord}\, f_*$ and   ${\rm Ord}\, f_*$ respectively.

Difference operators form a unital ring $\cR=\cF[\cS,\cS^{-1}]$ of 
Laurent 
polynomials in $\cS$ with coefficients in $\cF$, usual addition 
and multiplication defined by 
\begin{equation} \label{smult}
 a\cS^n \cdot b\cS^m=a\cS^n(b)\cS^{n+m}.
\end{equation}
This multiplication is associative, but non-commutative. Definitions of some 
basic concepts for a unital associative ring are presented in the Appendix A.

From the above definition it follows that if  $A$ is a difference operator of 
order ${\rm ord}\, A=(p,q)$, then ${\rm 
ord}\, (\cS^n\cdot  A\cdot \cS^m) =(p+n+m,q+n+m)$ and ${\rm Ord}\, (\cS^n\cdot  
A\cdot \cS^m )={\rm Ord}\, 
A=q-p$. For any $A,B\in\cR$ we have $\Ord\, (AB)=\Ord\, A+\Ord\, B$. Thus total 
order is homomorphisms of the multiplicative monoid $\cR$ to $\bbbz_{\ge 
0}\cup\{\infty\}$.

Reflection $\cT$ can be extended to automorphism of $\cR$ given by
\[
 \cT\cdot a\cS^m\cdot\cT=\cT(a)\cS^{-m}
\]
and define a grading of $\cR$ as follows:
\[
 \cR=\cR_0\oplus\cR_1,\qquad \cR_k=\{A\in\cR\,|\, \cT\cdot  A\cdot 
\cT=(-1)^k A\}.
\]
It is obvious that ${\rm Ord}(\cT\cdot A\cdot\cT)={\rm Ord}\, A$.

For a difference operator $B$ (\ref{operB})  the {\em leading 
monomial} $B_L$ is, by definition, $ B_L=b_N \cS^{N}$.

Monomial difference operators are of the form $a\cS^n,\ a\ne 0$. They have  
total order equal to zero and are
invertible in $\cR$. Monomial difference operators form a 
nonabelian group 
\[
 \cR^\times =\{a\cS^n\, |\, a\in\cF,\ a\ne 0,\ n\in \bbbz\}
\]
with multiplication (\ref{smult}).

\begin{Pro}
 The ring $\cR$ is a right and left Euclidean domain.
 \end{Pro}
\begin{proof}
Let us show that $\cR$ is a right  Euclidean, that 
is for any $A,B\in \cR$ there exist unique $Q, R\in 
\cR$ such that $A=B\cdot Q+R$ and either $R=0$ or $\Ord\, R<\Ord\, 
B$.
First we prove the existence of $Q,R$. If $A=0$, then we can take $Q=R=0$. If $A\ne 
0$ and $\Ord\, A<\Ord\, B$, we can take $Q=0,\ R=A$. For $\Ord\, A\ge \Ord\, B$ 
we proceed by induction on $\Ord\, A$. If $\Ord\, A=0\ (\Ord\, B)$, then 
$A=a\cS^N,\ B=b\cS^M$ for some $N,M\in\bbbz$ and they are invertible. Thus 
$A=BB^{-1}A$ and we can take $R=0, \ Q=B^{-1}A=\cS^{-M}(a/b)\cS^{N-M}$. 
Finally, 
consider the case $\Ord\, A=n\ge 1,\ \Ord\, B=m,\ n\ge m$ and assume that the 
statement is true for all operators $A$ with total order less than $n$. Let the 
leading monomials of $A$ and $B$ be $a\cS^N$ and  $b\cS^M$ respectively. The 
difference operator $\hat A=A-B \cdot (b\cS^M)^{-1}\cdot a\cS^N $ has  $\Ord\, 
\hat A<\Ord\, A=n$. Hence we can use the induction assumption and find $\hat 
Q,\ 
\hat R$, such that $\hat A=B\hat Q+\hat R$ and either $\hat R=0$ or $\Ord\, 
\hat 
R<\Ord B$. Thus
\[
 A-B 
\cdot (b\cS^M)^{-1}\cdot a\cS^N=B\hat Q+\hat R,
\]
that is,
\[
 A=B((b\cS^M)^{-1}\cdot a\cS^N+\hat Q)+\hat R.
\]
Therefore $Q=(b\cS^M)^{-1}\cdot a\cS^N+\hat Q$ and $R=\hat R$. As for the
uniqueness, if one has $BQ+R=B\tilde Q+\tilde R$ with $\Ord \, R<\Ord B,\ \Ord 
\, \tilde R<\Ord B$, then $B(Q-\tilde Q)=\tilde R-R$. If $Q\ne \tilde Q$ we 
arrive to a contradiction since $\Ord(B(Q-\tilde Q))>\Ord(\tilde R-R)$. Thus 
$Q= \tilde Q$ and $R= \tilde R$. The proof of the left Euclidean property is 
similar. 
 \end{proof}

\begin{Cor}\label{principal}
Every right (left) ideal of the ring $\cR$ is  principal and 
generated by a unique element $A\in \cR$ of minimal possible 
order with the leading monomial $A_L=1$.
\end{Cor}
\begin{proof}
The zero ideal is obviously principal, it is generated by 
$0$. Let $J\subset \cR$ be a right ideal and $\hat A\in J$ be an 
element of least possible total order. The element $A=\hat A\cdot \hat A_L^{-1}\in 
J$, where $\hat A_L$ is the leading monomial of $\hat A$, is of the same total 
order and 
with the leading monomial $A_L=1$. Then for any other element $B\in J$ we 
have $B=AQ+R$ with either $R=0$ or $\Ord\, R<\Ord\, A$. Since $B\in J$, 
we conclude that $R=0$, otherwise $\Ord\, R<\Ord\, A$, which is in contradiction 
with the assumption that $A$ has the least possible order. Such element $A$ is 
obviously unique. If we assume the existence of $\tilde A\in J, \ \Ord \,
\tilde A=\Ord\, A, \tilde A_L=1$, then $A-\tilde A\in  J$ and $\Ord \, 
(A-\tilde A)<\Ord\, A$. In a similar way we 
show that   $\cR$ is  a left principal ideal ring. 
\end{proof}

\begin{Pro}\label{Oreprop}
 The ring $\cR$ satisfies the right (left) Ore property, that is, 
for any $A,B\in \cR$ their exist $A_1,B_1$, not both equal to 
zero, such that $AB_1=BA_1$, (resp. $B_1 A=A_1 B$). In other words, the right (left) ideal
$A \cR \cap B \cR$ (resp. $ \cR A \cap  \cR B$) is nontrivial. Its generator $M$ has total order
$\Ord A+ \Ord B-\Ord D$, where $D$ is the greatest left (resp. right) common 
divisor of $A$ and $B$.
\end{Pro}
\begin{proof} 
Let us assume that $\Ord\, A\ge \Ord\, B$ (otherwise we swap and 
rename $A,B$).  If $B=0$,  then $B_1=0$. If $B\ne 0$, we prove the claim by 
induction on $\Ord\, B$. We assume that the statement 
is true for any $B$ with $\Ord\, B<k$ and we will show that it is also true for any $B,\ 
\Ord\,B=k$. Since $\cR$ is right Euclidean, there exist $Q,R$ such 
that $A=BQ+R$ and either $R=0$ or $\Ord\, R<\Ord B$. If $R=0$ we take $A_1=Q,\ 
B_1=1$ and we are done. Since  $\Ord\, R<k$, there exist $\hat B$, $\hat R$ such that 
$B \hat R=R \hat B$, $\Ord \hat R \leq \Ord R$ and $\Ord \hat B \leq \Ord B$. Thus 
\[
 A \hat B=(BQ+R)\hat B\ \Leftrightarrow\ A \hat B=B (Q \hat B+\hat R)
\]
and we can take $A_1=Q\hat B+ \hat R$,  $B_1=\hat B$. Finally, we can see that $\Ord A_1 \leq \Ord A$ and $\Ord B_1 \leq \Ord B$. The proof of the left Ore 
property is similar. 

We proved that for any $A, B \in \cR$ not both zero, the ideal $\mathcal{I}=A \cR \cap B \cR$  is not trivial. Since $\cR$ is both a right and left principal ideal ring, $\mathcal{I}$ is generated by a difference operator $M$, $\mathcal{I}=M \cR$. In particular, $M=AB_1=BA_1$ for some difference operators $B_1$ and $A_1$. From the first part of the proof, we know that  $\Ord M \leq \Ord A+ \Ord B$. Let us assume that $A$ and $B$ are left coprime and that $\Ord M< \Ord A+\Ord B$. The ideal $\mathcal{J}=\cR A_1 \cap \cR B_1$ is also non trivial and generated by a difference operator $N$. 
We know that $\Ord N$ is at most $\Ord A_1 + \Ord B_1$. $M$ is an element of $\mathcal{J}$ and $\Ord M=\Ord A + \Ord B_1> \Ord A_1+ \Ord B_1 \geq \Ord N$, hence there exists a difference operator $C$ such that $M=CN$ and $\Ord C >0$.  Let $A_2$ and $B_2$ be such that $A_2B_1=B_2A_1=N$. Then $A=CA_2$ and $B=CB_2$, which contradicts the hypothesis that $A$ and $B$ are left coprime.   
\end{proof}

The domain $\cR$ can be naturally embedded in the skew field of 
rational pseudo--difference operators, which we will call simply \textit{rational operators}.

\begin{Def}
 A rational (pseudo--difference) operator $L$ is defined as $L=AB^{-1}$ 
for some  $A,B\in\cR$ and $B\ne 0$. The set of all 
rational operators is
 \[
  \cQ=\{ AB^{-1}\,|\, A,B\in \cR,\ B\ne0\}.
 \]
\end{Def}

\begin{Rem}
The skew field  $\cQ$ is a minimal subfield of the skew field $\cQ^L$ of the 
Laurent formal series
\[
 \cQ^L=\left\{\sum_{n=k}^\infty a^{(-n)}\cS^{-n}\ |\  a^{(n)}\in\cF,\ 
n\in\bbbz\right\}
\]
containing $\cR$. As well as it is a minimal subfield of the skew field $\cQ^T$ 
of the Taylor
formal series
\[
 \cQ^T=\left\{\sum_{n=k}^\infty a^{(n)}\cS^{n}\ |\  a^{(n)}\in\cF,\ n\in\bbbz\right\}
\]
containing $\cR$. The skewfields $\cQ^L$ and $\cQ^T$ are isomorphic. The 
isomorphism is given by the reflection map $\cT$.
\end{Rem}
\begin{Pro}
Any rational operator 
$L=AB^{-1}$ can also be written in the form $L=\hat{B}^{-1}\hat A,\ \hat A,\hat 
B\in\cR$ and $\hat B\ne 0$. 
\end{Pro}
\begin{proof}
 Indeed, it follows from the Ore 
condition that for any $A,B\in\cR,\ B\ne 0$ there exist $\hat 
A,\hat 
B\in\cR$ and $\hat B\ne 0$ such that 
$\hat{B}A=\hat{A}B$. Multiplying this expression on $\hat B$ from the left and 
$B^{-1}$ from the right we obtain $L=AB^{-1}=\hat{B}^{-1}\hat A$. 
\end{proof}

Thus any statement for the representation $L=AB^{-1}$ can be easily 
reformulated to the representation $L=\hat B^{-1}\hat A$. In particular,
\[
  \cQ=\{ AB^{-1}\,|\, A,B\in \cR,\ B\ne0\}=\{ B^{-1}A\,|\, 
A,B\in \cR,\ B\ne0\}.
 \]
\begin{Pro}
 $\cQ$ is the skew field of rational operators over $\cF$.
\end{Pro}
\begin{proof} We need to show that the set  $\cQ$ is closed under addition and multiplication. Let 
$A,B,C,D\in \cR$ with $B\ne 0, D\ne 0$. It follows from the Ore 
property that there exist nonzero $\hat B,\hat D\in \cR$ such that 
$B\hat D=D\hat B$. Hence 
\[
 AB^{-1}+CD^{-1}=(A\hat D+C\hat B)\cdot (B\hat D)^{-1}\in \cQ .
\]
Also there exist nonzero $\hat B,\hat C$ such that $B\hat C=C\hat B$. Hence
\[
 (AB^{-1})\cdot (CD^{-1})=(A\hat C)\cdot (D\hat B)^{-1} \in \cQ
\]
implying $\cQ$ is also closed under multiplication. 
\end{proof}
\begin{Pro}\label{pro52}
 The decomposition $L=AB^{-1},\ A,B\in \cR$ of an element $L\in 
\cQ$ is unique if we require that $B$ has a minimal possible total order 
with leading monomial $B_L=1$. For any other decomposition $L=\hat 
A\hat B^{-1},\ \hat A,\hat B\in \cR$ there exists $C\in \cR$ such 
that $\hat A=AC,\ \hat B=BC$. Moreover, if $D^{-1}E$ is a (left) minimal decomposition of 
$L$, then $\Ord D=\Ord B$.
\end{Pro}
\begin{proof}
For a given $L\in\cQ$ the set 
\[
 J=\{X \in  \cR\, |\, LX \in \cR\}
\]
is a right ideal in $\cR$. Indeed, if $X, Y \in J$, then $L(X+Y)=LX+LY \in \cR$ meaning that $X+Y \in J$, and $J$ is stable under right multiplication by any element of $\cR$. The ideal $J$ is principal, and according to Corollary \ref{principal} 
it is generated by a unique element $B$ of the least possible order, if we 
require that the leading monomial $B_L=1$. Any other $\hat B\in J$ can be 
represented as $\hat B=BC$ where $C\in\cR$, since $B$ is a generator of 
the principal right ideal $J$. By Proposition \ref{Oreprop}, we know that the generator $M$ of the left ideal generated by 
$A$ and $B$ has total order $\Ord A+ \Ord B$. By definition of $M$ there exist left coprime difference operators $D$ and $E$ such that $DA=EB=M$. Therefore $D^{-1}E$ is a left minimal decomposition of $L$ and $\Ord D=\Ord B$.
\end{proof}

The definition of total order for difference operators (Definition 
\ref{deford}) can be extended to rational operators:
\begin{equation}\label{ordF}
 \Ord\, (AB^{-1}):=\Ord\,A-\Ord\, B,\quad 
 \Ord\, (\hat B^{-1}\hat A):=\Ord\,\hat A-\Ord\,\hat B,\qquad A,B,\hat A, \hat 
B\in \cR. 
\end{equation}

\begin{Def}
A formal adjoint operator $A^\dagger$ for any $A\in\cQ$ can be 
defined 
recursively:
\begin{enumerate}
 \item $a^\dagger=a$ for any $a\in\cF$,
 \item $\cS^\dagger=\cS^{-1}$,
 \item $(A\cdot B)^\dagger=B^\dagger \cdot A^\dagger$ for any $A,B\in\cQ$,
 \item $(A^{-1})^\dagger=(A^\dagger)^{-1}$  for any $A\in\cQ$.
\end{enumerate}
\end{Def}
In particular, We say an operator $H\in \cQ$ is anti-symmetric if $H^{\dagger}=-H$.

For example, we have 
\[
\left( (\cS+a\cS^{-1})\cdot (b-\cS)^{-1}\right)^\dagger=(b-\cS^{-1})^{-1}\cdot 
(\cS^{-1}+\cS(a)\cS),\qquad a,b\in\cF.
\]
For any $A\in \cQ$,  if ${\rm ord}\, A=(p, q)$ then ${\rm ord}\,A^\dagger=(-q, -p)$.   Obviously $\Ord\, A^\dagger=\Ord\, A$.

\subsection{Rational and weakly non-local difference operators}\label{sec23}
In the theory of integrable systems,  the majority  of $1+1$-dimensional integrable equations possess 
weakly nonlocal \cite{MaN01} Nijenhuis recursion operators. For integrable 
differential-difference equations, weakly 
non-local operators are often rational pseudo--differential operators with only a finite
number of nonlocal terms of the form $a (\cS-1)^{-1}b$, where $a, b \in \cF$. 
In 
this section, we answer the questions 
how to write a weakly 
nonlocal operator as a rational operator and when a rational difference 
operator 
is indeed weakly nonlocal. For the
differential case, the answers are given by Lemma $4.5$ in \cite{Carp2017}.

First we give a definition of the full kernel operators. We then prove that for 
such operators, their inverse operators 
are weakly nonlocal.

For a difference operator $A\in\cR$ it is obvious that 
\begin{equation}\label{kerord}
 \mbox{dim}_\cK 
\mbox{Ker}\,A\le \Ord\, A.
\end{equation}
 Indeed, if there is an element $a\in\cF$ 
such that $a\in \mbox{Ker}\,A$, then we can represent $A=\tilde A
(\cS-1) \frac{1}{a},$ where $\Ord\,\tilde A=\Ord\,A-1$. Zero total 
order difference operator is invertible and thus it has a trivial kernel space. 
A difference operator of a non-zero order may also have a trivial kernel in 
$\cF$ as well. For example 
$\mbox{Ker}_\cK(\cS-u)=0$ since equation $\cS(v)=uv$ does not have a solution 
$v\in\cF$. 

\begin{Def}
 We say that a difference operator has a full kernel in $\cF$ (is a full kernel 
operator) if the dimension of its kernel 
over 
the field $\cK$ equals to the total order of the operator.
\end{Def}
In what follows, we show how to construct a full kernel operator given the generators of its kernel and prove
an important property of such operators.
\begin{Pro}\label{profk}
Assume that $f^{(1)}, \cdots, f^{(n)}$ are linearly independent over $\cK$ in $\cF$.
Then there exists a full kernel difference operator $P\in\cR$ such that the $f^{(i)}, i=1, \cdots, n$ span 
$\ker P$.
\end{Pro}
\begin{proof}
We prove the statement by induction on $n$. If $n=1$, we define
$$P=(\cS-1)\frac{1}{f^{(1)}}.$$
It is clear that ${\Ord}\,P=1$ and its kernel is spanned by $f^{(1)}$. 
Assume that $Q$ is a full kernel operator with $\Ord\,Q=n-1$ and its kernel is spanned by $f^{(1)}, \cdots, f^{(n-1)}$.
Since $f^{(i)}, i=1, \cdots, n$ are linearly independent, we have $Q(f^{(n)})\neq 0$ by construction of $Q$.
We define
$$P=(\cS-1)\frac{1}{Q(f^{(n)})} Q.
$$
Clearly it is the required full kernel operator and its kernel is spanned by $f^{(1)}, \cdots, f^{(n)}$.
\end{proof}
\begin{Rem}
A difference operator $Q\in \cR$ with the full kernel spanned by 
$\cK$--linearly independent elements $f^{(i)}\in\cF, i=1, \cdots, n$, can be 
obtained using the determinant expression
\[
 Q(g)=\det\left(\begin{array}{lcll}
    f^{(1)}&\cdotp&   f^{(1)}& g\\  
     \cS( f^{(1)})&\cdotp&  \cS( f^{(1)})&\cS( g)\\  
    \vdots&\cdots&×\cdotp&\cdotp\\
  \cS^n( f^{(1)})&\cdots&  \cS^n( f^{(1)})&\cS^n( g)   
                \end{array}
\right)\ \mbox{for and $g\in \cF$.}
\]
\end{Rem}

\begin{Pro}\label{prop41}
 The inverse operators of full kernel operators are weakly nonlocal.
\end{Pro}
\begin{proof} We prove the statement by induced on the total order of such 
operator
$B$. If  $B$ is a full kernel operator with $\Ord B=1 $, it can be written as 
$B=a \cS^i (\cS-1) b$ for some $i\in \mathbb{Z}$. Thus 
$$B^{-1}= \frac{1}{b} \cS^{-i} (\cS-1)^{-1} \frac{1}{a}$$ 
is weakly nonlocal. 

Let $B$ be a full kernel operator with the total order of $n$ and $a\in \ker 
B$. It follows from Proposition \ref{profk} that
there is a full kernel operator $C$ 
with total order of $n-1$ such that
$$B=C (\cS-1)\frac{1}{a}.$$
By the induction assumption, $C^{-1}$ is weakly nonlocal, that is, there exist 
two sets of linearly independent functions $b^{(i)}$ and $c^{(i)}$,  
$i=1,\cdots, 
n-1$
such that
$$C^{-1}=E+\sum_{i=1}^{n-1} b^{(i)}  (\cS-1)^{-1} c^{(i)}, \quad E\in \cR .$$
Multiplying $C$ on its left, we get
$$\sum_{i=1}^{n-1} C(b^{(i)})  (\cS-1)^{-1} c^{(i)}=0 $$
implying $b^{(i)}\in \ker C. $
Note that for any $b^{(i)}\in \ker C$, $i=1,\cdots, n-1$, there exists 
$d^{(i)}$, which is in $\ker B$
such that
$b^{(i)}=(\cS-1)\frac{d^{(i)}}{a}.$ Therefore, we have
\begin{eqnarray*}
 B^{-1}= a  (\cS-1)^{-1} C^{-1}= a  (\cS-1)^{-1} \left(E+\sum_{i=1}^{n-1} 
\left((\cS-1)\frac{d^{(i)}}{a}\right)  
(\cS-1)^{-1} c^{(i)}\right),
\end{eqnarray*}
whose nonlocal terms are 
$$a  (\cS-1)^{-1} E^{\dagger}(1) +\sum_{i=1}^{n-1} \left( d^{(i)} (\cS-1)^{-1} 
c^{(i)}-a (\cS-1)^{-1} 
\frac{c^{(i)} d_1^{(i)}}{a_1} \right), $$
where we used the identity 
$$(\cS-1)^{-1}  (d_1-d)  (\cS-1)^{-1} =d (\cS-1)^{-1} -(\cS-1)^{-1} d_1, \quad 
d\in \cF.$$
This leads to the conclusion that $B^{-1}$ is weakly non-local.
\end{proof}
We are now ready to prove the statement on the relation between the rational 
and 
weakly nonlocal difference operators.
\begin{The}\label{thmrw}
Let $R=A B^{-1}$ be a rational operator with minimal right fractional 
decomposition and ${\Ord B}=n$. Then the following 
three statements are equivalent:
\begin{enumerate}
 \item[{\rm (i)}]  Operator $B$ has a full kernel in $\cF$;
 \item[{\rm (ii)}] Operator $R$ is weakly nonlocal, that is,  $R=L+\sum_{i=1}^n 
p^{(i)} (\cS-1)^{-1} q^{(i)}$, where $L\in\cR$, and $\{p^{(i)}, i=1, \cdots, n\}$ and 
$\{q^{(i)},i=1,\cdots, n\}$ are two linearly independent sets
over $\cK$ in $\cF$;
 \item[{\rm (iii)}] Operator $B^{\dagger}$ has a full kernel in $\cF$.
\end{enumerate}
\end{The}
\begin{proof} The statement ${\rm (i)}\Rightarrow {\rm (ii)}$ directly follows 
from Proposition \ref{prop41} since the 
multiplication of a difference operator and a weakly nonlocal operator is 
weakly 
nonlocal. 

We now prove that ${\rm (ii)}\Rightarrow {\rm (iii)}$. Knowing
$$R=A B^{-1}=L+\sum_{i=1}^n p^{(i)}(\cS-1)^{-1} q^{(i)}, $$
we multiply it on the right by $B$ and obtain its nonlocal terms
$$
\sum_{i=1}^n p^{(i)}(\cS-1)^{-1} B^{\dagger}(q^{(i)})=0,
$$
which implies that all $q^{(i)}$'s are in the kernel of $B^{\dagger}$ and thus 
$n\leq \dim (\ker B^{\dagger})$.

Let $C$ be a common multiple of the difference operators $\frac{1}{q^{(i)}} 
(\cS-1)$, that is, a difference operator such that for all $i$ there exists a 
difference operator
$M^{(i)}$ satisfying $C=\frac{1}{q^{(i)}} (\cS-1) M^{(i)}$. Thus we have
$$R=L+\sum_{i=1}^n p^{(i)}(\cS-1)^{-1} q^{(i)}=(LC+\sum_{i=1}^n p^{(i)}M^{(i)})C^{-1}. 
$$
Since $AB^{-1}$ is a minimal right fractional decomposition for $R$, there 
exists a difference operator $D$ such that
$$
LC+\sum_{i=1}^n p^{(i)}M^{(i)}=AD \quad \mbox{and} \quad C=BD.
$$
This leads to $\Ord C=n\geq \Ord B$. Note that $\Ord B=\Ord B^{\dagger}$ and 
$\Ord B^{\dagger}\geq \dim (\ker 
B^{\dagger})$. Therefore, we have
$$
\Ord B^{\dagger}=\dim (\ker B^{\dagger})=n
$$
implying that $B^{\dagger}$ has a full kernel spanned by all $q^{(i)}$'s.

Finally we prove that ${\rm (iii)}\Rightarrow {\rm (i)}$. It follows from 
Proposition \ref{prop41} that $B^{\dagger}$ 
is weakly nonlocal. Using the proof of ${\rm (ii)}\Rightarrow {\rm (iii)}$, we 
obtain that statement of ${\rm (i)}$.
\end{proof}
From the proof of Theorem \ref{thmrw}, we are able to specify the nonlocal terms for weakly nonlocal operator.
\begin{Cor}\label{cork} Under the condition of Theorem \ref{thmrw}, for $R=L+\sum_{i=1}^n 
p^{(i)} 
(\cS-1)^{-1} q^{(i)}$, the linearly independent functions $p^{(i)}$'s 
span $A(\ker B)$ and the linearly independent functions
$q^{(i)}$'s span $\ker 
B^\dagger$,  $i=1, 
\cdots, n$.
\end{Cor}

Following from this theorem, we are immediately able to get the statement for the 
inverse of rational operator:
\begin{Cor}\label{inverse}
Let $R=AB^{-1}$ with $A, B\in \cR$. Then $R^{-1}$ is weakly non-local if and 
only if $A$ has a full kernel in $\cF$. 
\end{Cor}
Corollary \ref{cork} combined with Proposition \ref{profk} provides us a method to write a weakly nonlocal operator in 
the 
form of a rational operator $R=AB^{-1}$: We first construct a full kernel operator $B^\dagger$ using $q^{(i)}$'s. Then we 
have $A=RB$. We use such construction for the examples in Section \ref{Sec6}, where
we'll also apply Corollary \ref{inverse} to the recursion operators of 
integrable differential-difference 
equations to see whether their inverse operators are weakly nonlocal or not. If 
so, we are going to compute the 
seeds for symmetry and co--symmetry hierarchies (its nonlocal terms), that is, 
the $p^{(i)}$'s and $q^{(i)}$'s for $R^{-1}$ in the above theorem.

\subsection{Matrix difference and rational pseudo--difference operators}

We recall here some facts from linear algebra over  non-commutative rings and 
skew fields, which is a specialisation of the  
general theory \cite{Artin57, draxl_1983} to the case of difference 
algebra (the ring 
$\cR$ and skew field  $\cQ$). We denote $\cM_n(\cR)$ and $\cM_n(\cQ)$ the rings 
of $n\times n$ matrices over the ring $\cR$ and skew field $\cQ$ respectively. 
Since $\cR$ is a  
principal ideal ring, then the ring $\cM_n(\cR)$ is also a  principal ideal ring
(see proof in \cite{MR}, as well a short and useful review of 
non-commutative 
principal ideal rings one can find in \cite{CDSK2013}).

 Let $\cA_i$ denotes 
the $i$--th row of the matrix $\cA$ and $\cA_{i,j}$ denotes the $(i,j)$ entry 
of $\cA$. For $1\le i\ne j\le n$ and arbitrary $B\in\cR$ (or $B\in\cQ$) the 
$\cR$--elementary (resp. $\cQ$--elementary) row operation  $\tau_{i,j}(B)$ 
changes the row 
$\cA_i\mapsto \cA_i +B\cdot \cA_j$ and 
leaves the other rows unchanged. Transformation $\tau_{i,j}(B)$ is invertible 
$\tau_{i,j}(B)\tau_{i,j}(-B)={\rm Id}$ and can be represented by a 
multiplication from the left by the 
 matrix $ \tau_{i,j}(B)=I+B E_{i,j}$, where $I$ is the unit 
 matrix and $E_{i,j}$ is the matrix with the  $(i,j)$ 
entry equal to $1$ and zero elsewhere. Note that the transformation 
$\sigma_{i,j}=\tau_{i,j}(1) 
\tau_{j,i}(-1)\tau_{i,j}(1)$ replaces $\cA_i$ by $\cA_j$ and $\cA_j$ by 
$-\cA_i$, leaving other rows unchanged.

$\cR$--elementary row operations 
generate  a group $\cE_n(\cR)$, which is a subgroup of the group $GL_n(\cR)$ of 
invertible matrix difference operators. Respectively, $\cQ$--elementary row 
operations 
generate  a group $\cE_n(\cQ)$,  a subgroup of the group $GL_n(\cQ)$ of 
invertible matrix pseudo--difference operators. 
\begin{Lem}\label{lem1}
Let $\cA\in\cM_n(\cR)$. Then there exist two invertible matrices $\mathcal{U}$ and $\mathcal{V}$ such that $\mathcal{U}\cA \mathcal{V}$ is diagonal.
\end{Lem}
\begin{proof}
Let $\mathcal{N}$ be an element of the set $E=\{ \mathcal{U}\cA\mathcal{V} |  \mathcal{U}, \mathcal{V} \in \cE_n(\cR)\}$ such that for all $\cM \in E$, either $\cM_{11}=0$ or $\Ord \, \mathcal{N}_{11} \leq \Ord \, \cM_{11}$. 
We claim that all entries in the first column of $\mathcal{N}$ are divisible on the right by $\mathcal{N}_{11}$. Otherwise, using elementary row operations which amounts to multiply $\mathcal{N}$ on the left by an invertible matrix, 
one can find $\cM\in E$ such that $\cM_{11} \neq 0$ and $\Ord \, \cM_{11}< \Ord \, \mathcal{N}_{11}$, which contradicts 
the definition of $\mathcal{N}$. Similarly, $\mathcal{N}_{11}$ must divide all the entries of the first row of 
$\mathcal{N}$ on the left. Therefore, there exist invertible matrix difference operators $\mathcal{U}$ and $\mathcal{V}$ 
such that $\mathcal{U}\mathcal{N}\mathcal{V}$ has only zero entries in its first row and first column, apart from the 
first coefficient which is $\mathcal{N}_{11}$. We conclude by induction on $n$.
\end{proof}
\begin{Pro}\label{echelon}
 Let $\cA\in\cM_n(\cR)$. Then  it can be 
brought to 
a upper triangular form $ \cA^\vartriangle$ with  
$\cA_{i,j}^\vartriangle=0$ for $ i>j$ by $\cR$--elementary row operations and
\[
 \cA^\vartriangle=\cG\cA,\quad \cG\in \cE_n(\cR).
\]
\end{Pro}

\begin{proof} We prove the claim by induction on $n$. If $n=1$, the matrix is 
already in the form required. Now we assume that any matrix from 
$\cM_{n-1}(\cR)$ can be brought to a upper triangular form by $\cR$--elementary 
row transformation. Therefore the first $n-1$ rows of matrix $\cA$ can be 
brought 
to the upper triangular form. 

(i) If  $\cA_{n,1}=0$, then deleting the first row and first column  
of $\cA$ we reduce 
the problem to the case $\cM_{n-1}(\cR)$ and we are done 
due to the induction conjecture.

(ii) If $\cA_{1,1}=0$, then we use transformation $\sigma_{1,n}$ to reduce 
the 
problem to the  case (i). 

(iii) The remaining case is $ \cA_{1,1}\ne 0,\ 
\cA_{n,1}\ne 0$. Let $\Ord \, \cA_{n,1}\le \Ord \, \cA_{1,1}$ (if otherwise, we 
can swap the rows by the transformation $\sigma_{1,n}$), then there are exist
$B,R\in \cR$ such that $\cA_{1,1}=B\cdot \cA_{n,1}+R$ and either $R=0$ or 
$\Ord\, R<\Ord \, \cA_{n,1}$ and we apply the transformation $\tau_{1,n}(-B)$ 
replacing $\cA_1$ by $\hat\cA_1=\cA_1-B\cA_n$.
If $R=0$, then the updated row $\hat \cA_{1}$ has zero entry $\hat \cA_{1,1}=0$ 
and we are done (ii), or $\Ord \, \hat \cA_{1,1}< \Ord \,  \cA_{n,1}$ 
and we use $\sigma_{1,n}$ to swap the rows. Iterating 
this procedure we can vanish the entry $(n,1)$,  reducing 
the problem to the case (i).
\end{proof}

The ring $\cM_{n}(\cR), n\ge 2$ has zero divisors. The multiplicative monoid of
regular elements, i.e. the elements which are not zero divisors, we will denote 
$\cM_{n}^\times(\cR)$. A difference matrix operator  $\cA$ is \textit{regular}, if and 
only if  its upper triangular form 
$\cA^\vartriangle$ is regular, i.e. $\cA^\vartriangle$ has a non-zero product 
of its diagonal elements 
\begin{equation}\label{delepta}
  \delta\epsilon\tau (\cA^\vartriangle ):=\cA_{1,1}^\vartriangle \cdot 
\cA_{2,2}^\vartriangle\cdots \cA_{n,n}^\vartriangle\ne 0.
\end{equation}

\begin{Def}\label{rordm} The total order of a  matrix difference 
operator 
$\cA\in 
\cM_{n}(\cR)$ is defined as the sum of total orders of the 
diagonal entries of a corresponding upper triangular operator 
$\cA^\vartriangle$
\[
 \Ord\, \cA=\sum_{i=1}^n \Ord\, \cA^\vartriangle_{i,i}=\Ord\, \delta\epsilon\tau 
(\cA ) .
\]
\end{Def}

\begin{Pro}\label{invert}
 A difference matrix operator $\cA$ is invertible in $\cM_{n}(\cR)$  (i.e. 
$\cA^{-1}\in 
\cM_{n}(\cR)$ and thus
$\cA\in GL_n(\cR)$), if and only if $\Ord\, \cA=0$.
\end{Pro}

\begin{proof} If  $\Ord\, \cA=0$, then all entries on the diagonal part 
$\cA^\vartriangle_d:={\rm 
diag}((\cA_{1,1}^\vartriangle) ,\ldots, (\cA_{n,n}^\vartriangle))$ 
of $\cA^\vartriangle$ have total order zero and thus invertible. 
Multiplying $\cA^\vartriangle$ on the left by matrix 
$(\cA^\vartriangle_d)^{-1}$ we obtain an upper triangular matrix 
$\tilde{\cA}^\vartriangle= (\cA^\vartriangle_d)^{-1}\cG\cA$ with the unit 
matrix on the diagonal. By induction on $n$ it is easy to show that there is a 
composition of $\cR$--elementary row transformations $\tilde\cG$ such that 
$\tilde\cG \tilde{\cA}^\vartriangle=I$. If $n=1$ it is nothing to do. We 
assume the existence of the inverse matrix in $\cM_{n-1}(\cR)$. The entries 
$\tilde{\cA}^\vartriangle_{k,n},\ k=1,\ldots,n-1$ of last column can be set to 
zero by the transformation 
$\prod_{k=1}^{n-1}\tau_{k,n}(-\tilde{\cA}^\vartriangle_{k,n}) \cdot 
\tilde{\cA}^\vartriangle$ which reduces the problem to the case in 
$\cM_{n-1}(\cR)$. The necessity is obvious from the consideration of a diagonal 
matrix
$\cA$.
\end{proof}
\begin{Ex}
Let us consider the following matrix difference operator
\begin{equation}\label{exampla}
 \cA=\left(\begin{array}{cc}
1&\cS^2\\a_{-1}\cS^{-1}&a_1\cS            ×
           \end{array}\right),
\end{equation}
where $a\in\cF,\ a_k=\cS^k(a)$ and $a_i\ne a_{j}$ if $i\ne j$. Transformation 
$\cA\mapsto\cA^\vartriangle=\tau_{2,1}(-a_{-1}\cS^{-1})\cA$ brings $\cA$ in the 
upper triangular form and
$ \delta\epsilon\tau (\cA^\vartriangle )=(a_1-a_{-1})\cS$. Thus $\Ord \, \cA=0$ 
and there exists the inverse matrix difference operator. Indeed,
 \[
 \cA^{-1}=\left(\begin{array}{ccc}
\dfrac{a_2}{a_2-a}&\ &-\dfrac{1}{a_2-a}\cS\\& &\\
-\dfrac{a_{-2}}{a-a_{-2}}\cS^{-2}&\ &\dfrac{1}{a-a_{-2}}\cS^{-1}          ×
           \end{array}\right).
\]
If we use a different sequence of elementary row transformations 
\[ \cA\mapsto\tilde 
\cA^\vartriangle=\tau_{1,2}(-\frac{1}{a_{2}}\cS)\tau_{2,1}(-\frac{a_1 
a_{-1}}{a_1-a_{-1}}\cS^{ -1 } )\cA,\] 
which also brings the difference matrix operator $\cA $ 
to a upper triangular form, then $ \delta\epsilon\tau (\tilde\cA^\vartriangle 
)=\frac{a_1(a_2-a)}{a_2}\cS$, but the total order of $\cA$ does not depend on 
the choice of the sequence $\Ord\, A=\Ord\, \delta\epsilon\tau 
(\tilde\cA^\vartriangle 
)=0$ (see below).
\end{Ex}

The correctness of  Definition \ref{rordm}, i.e. the independence of 
$\Ord\,\cA$ 
from the choice of row transformations, can be justified by the theory of 
Dieudonn\'e determinants ($\Delta_n$) (in the case of skew polynomial rings 
it has been discussed in \cite{Taelman}). 
The above definition of total order for matrix difference operators is a 
restriction of the map $\Ord\,\Delta_n\,:\cM_n(\cQ)\mapsto \bbbz\cup\{\infty\}$ 
to $\Ord\,\Delta_n\,:\cM_n(\cR)\mapsto \bbbz_{\ge 0}\cup\{\infty\}$.  This 
observation results in a simpler way to compute the total order of matrix 
difference operators by treating them as elements of $\cM_n(\cQ)$.

Dieudonn\'e 
determinant $\Delta_n$ is defined for matrices with entries in an arbitrary 
skew field $\bbbk$ (see 
\cite{Artin57, draxl_1983, Dieu}). In our case the skew field is 
 $\bbbk=\cQ$ and we are dealing with matrix rational pseudo-difference operators
$\cM_n(\cQ)$, but what is presented below is equally applicable to 
rational pseudo--differential operators or any skew field of fraction of a left 
principal 
ideal domain. The Dieudonn\'e 
determinant is a map from $\cM_{n}(\cQ)$ to $\bar \cQ=\cQ^\times\diagup 
\cQ^{ 
(1)}$ or zero, where  $\cQ^\times$ is a multiplicative group of nonzero 
elements 
of $\cQ$, and $\cQ^{(1)}$ denotes the commutator subgroup 
$\cQ^{(1)}=[\cQ^\times,\cQ^\times]\subset \cQ^\times$, which is a 
normal. The 
group $\cQ^{ (1)}$ is generated by elements of the form 
$ABA^{-1}B^{-1},\ A,B\in\cQ^\times$. The quotient 
group $\bar \cQ$ is commutative and its elements are cosets 
$A\cQ^{ (1)},\ A\in \cQ^{\times}$. There is a natural projection $\pi\, :
\, \cQ^\times\mapsto \bar\cQ$ given by $\pi (A)=\bar A:=A\cQ^{ (1)}$ for 
any 
$A\in \cQ^\times$.

Dieudonn\'e has shown that $\cE_n(\cQ)$ is a normal subgroup  of
$  GL_n(\cQ)$ and  that there is a  
group isomorphism $\Delta_n: GL_n(\cQ)\diagup \cE_n(\cQ)\mapsto \bar \cQ$ 
given by a map $\Delta_n$ (Theorem 1. in \cite{Dieu}), which is now called the 
Dieudonn\'e determinant. The function $\Delta_n : \cM_n(\cQ)\mapsto\bar\cQ$ is: 
\begin{enumerate}
 \item multiplicative: $\Delta_n (\cA\cB)=\Delta_n(\cA)\Delta_n(\cB)$;
\item if $\cA\in\cE_n(\cQ)$, then $\Delta_n(\cA)=\bar 1$;
\item if $\cA'$ is obtained from $\cA$
by multiplying one row of $\cA$
on the left by $B\in \cQ$, then
\[
 \Delta_n \cA'=\bar B\,\cdot\,\Delta_n\cA;
\]
\item if  matrix $\cA$  is degenerate (i.e. one row is a left $\cQ$--linear 
combination of other rows), then $\Delta_n(\cA)=0$.
\end{enumerate}
In order to find $\Delta_n \cA$ for $\cA\in\cM_n(\cQ)$ one can use the 
algorithm given by Dieudonn\'e \cite{Dieu} (see also \S 1, Ch. IV 
\cite{Artin57}), or use the Bruhat normal form approach (\S 20, Part III,  
\cite{draxl_1983}). A simple way to find the Dieudonn\'e  determinant of a 
matrix $\cA\in\cM_n(\cQ)$ is to use 
a composition of $\cQ$--elementary row transformations in order to bring the 
matrix $\cA$ to a upper triangular form
$\cA^\vartriangle=\cG\cA,\ \cG\in\cE_n(\cQ)$, then multiply the diagonal 
entries of $\cA^\vartriangle$ (in an arbitrary order) and apply the projection 
$\pi$ to the result
\[
 \Delta_n(\cA)=\pi(\prod_{k=1}^n \cA^\vartriangle_{k,k}).
\]
It follows from \cite{Dieu} that  $\Delta_n(\cA)$ does not depend on the choice 
of elementary row transformations, neither on the order in the product of  
diagonal elements of $\cA^\vartriangle$.

It follows from Definition 
\ref{deford} and (\ref{ordF}) that $\Ord P=0$ for any $P\in \cQ^{ (1)}$, 
thus function $\Ord $ has a constant value on a coset and the map 
\[
 \Ord\, :\, \bar \cQ \ \mapsto\ \bbbz 
\]
is defined correctly.
\begin{Def}\label{qordm}
 The total order of a matrix rational operator $\cA\in\cM_n(\cQ)$ is 
\[
 \Ord\,\cA:=\Ord\,\Delta_n(\cA).
\]
\end{Def}

In the case of difference operators $\cA\in \cM_n(\cR)$ we have defined a 
function $\delta\epsilon\tau(\cA^\vartriangle)\in \cR$ (\ref{delepta}). 
Although the value of 
this function depends on the choice of $\cR$--elementary row transformations, 
but its natural projection to $\bar\cQ$ does not, since it coincides with the 
Dieudonn\'e determinant
\[
 \pi (\delta\epsilon\tau(\cA^\vartriangle))=\Delta_n(\cA).
\]

This restriction of the total order definition to the ring of matrix 
difference operators together with Proposition \ref{invert} results in the 
exact sequence of monoid homomorphisms (similar to Theorem 1.1 in 
\cite{Taelman}):
\[
 1\ \longmapsto\ GL_n(\cR)\ \longmapsto\ \cM_n(\cR)\ \longmapsto\ \bbbz_{\ge 
0}\cup\{\infty\}\ \longmapsto\ 0\ .
\]
Definition \ref{rordm} is a way to define the total order of a matrix 
difference operator, bypassing the  skew field of  rational 
operators, its quotient group $\bar \cQ$ and the theory of Dieudonn\'e 
determinants. 

Note that the Dieudonn\'e determinant and the total order of a matrix 
(rational) difference operator and the transposed matrix operator may not 
coincide. In the above example 
(\ref{exampla}): 
\[
 \Delta_2(\cA)=\pi((a_1-a_{-1})\cS),\qquad 
\Delta_2(\cA^{tr})=0.
\]

A formally conjugated  matrix 
(rational) difference operator has a 
usual definition, i.e. the corresponding matrix is transposed and each entry is 
formally conjugated: $(\cA^\dagger)_{i,j}=(\cA_{j,i})^\dagger$. For formally 
conjugated operators we have 
$\Delta_n(\cA^\dagger)=(\Delta_n \cA)^\dagger$ and therefore $\Ord 
\cA^\dagger=\Ord \cA$.

There are many ways to represent a matrix rational operator as a 
ratio 
of matrix difference operators. For example any $\cL\in\cM_n(\cQ)$ can be represented 
as
\[
 \cL=\hat \cA\cdot \cD^{-1}=\tilde\cA\cdot M^{-1},\quad \cD=\mbox{diag} 
(M_1,\ldots,M_n),\ M_k,M\in\cR\setminus\{0\}.
\]
Indeed, the entries  $\cL_{i,j}\in\cQ$ and thus 
$\cL_{i,j}=A_{i,j}B_{i,j}^{-1},\ A_{i,j},B_{i,j}\in\cR$. Since ring $\cR$ 
satisfies the Ore property (Proposition \ref{Oreprop}) there exist the least 
right common multiplier  $M_i$  of elements $B_{1,i},\ldots B_{n,i}$ and 
therefore there exist $P_{1,i},\ldots P_{n,i}\in \cR$ such that 
$M_i=B_{1,i}P_{1,i}=\cdots=B_{n,i}P_{n,i}$. Taking 
$\hat\cA_{i,j}=\cL_{i,j}P_{i,j}$ we obtain the first representation. 
Let $M$ be 
the least right common multiplier of $M_1,\ldots,M_n$ and therefore there exist 
$Q_1,\dots Q_n\in\cR$  such that $M=M_1 Q_1=\cdots =M_nQ_n$, then $\tilde A=\hat 
A\cdot \mbox{diag}(Q_1,\ldots,Q_n)$.

Since the ring of difference operators $\cR$ is principal ideal domain, the 
ring of matrices $\cM_n(\cR)$ satisfies the left and right Ore 
property (see proof in \cite{MR})
and thus
\begin{equation*}
\begin{split}
\mathcal{M}_{n}(\cQ) =&\{\cA\cB^{-1}|(\cA,\cB) \in 
\mathcal{M}_{n}(\cR) \times 
\mathcal{M}_{n}^{\times}(\cR)\} \\
 =&\{\cB^{-1}\cA| (\cA,\cB) \in \mathcal{M}_{n}(\cR) \times 
\mathcal{M}_{n}^{\times}(\cR)\} 
 \end{split}
\end{equation*}
A representation of matrix rational operators as right (left) 
fractions is not unique. However, once we clear the 
common right (resp. left) divisors, we get a 
minimal fraction, in the following sense:
\begin{The}
For any $\cL \in \mathcal{M}_{n}(\cQ)$ there is a minimal right 
(resp. left) decomposition $\cL=\cA\cB^{-1}$ (resp. $\cL=\hat\cB^{-1}\hat\cA$) 
with $\cA,\cB$ right 
(resp.  $\hat \cA,\hat \cB$ left) coprime. Any other right  
decomposition  $\cL=\cA_1\cB_1^{-1}$ (resp. left 
decomposition $\cL=\hat\cA_1^{-1}\hat\cB_1$) 
is of the  from $\cA_1=\cA\cdot \cC,\ \cB_1=\cB\cdot \cC$ (resp. 
$\hat\cA_1=\cC\cdot \hat\cA,\ \cB_1=\cC\cdot\hat \cB$), where 
$\cC\in\cM_n^\times(\cR)$. Moreover $\Ord\, \cB=\Ord\,\hat\cB$ and is minimal 
possible  among all decompositions.
\end{The}
\begin{proof}
We will first prove by induction on $n$ that
if $A$ and $B$ are matrix difference operators of size $n \times n$ with $B$ 
regular, if $M$ is a generator of the right ideal $A\cM_n(\cR) \cap B
\cM_n(\cR)$ and $N$ a greatest left common divisor of $A$ and $B$, then $\Ord A+ 
\Ord B=\Ord M+\Ord N$.
\\ It is true for $n=1$ by Proposition \ref{Oreprop}. 
Let us now consider $A$ and $B$ of size $n+1$. Using invertible matrices we can 
assume that $A$ and $B$ are both upper 
triangular. Indeed, one can factorize them as $A=T_AU_A$ and $B=T_BU_B$ with 
$T_A,T_B$ upper triangular and $U_A$, $U_B$ 
invertible. Hence if there exist $C$ and $D$ such that $T_AD=T_BC$ with $\Ord D 
\leq \Ord T_B=\Ord B$, then we can write 
$A(U_A^{-1}D)=B(U_B^{-1}C)$. Let us consider $A$ and $B$ in block matrix form:
\begin{equation*}
A=
\left(
\begin{matrix}
E & X \\
0 & P
\end{matrix}
\right),
\hspace{2 mm}
B=
\left(
\begin{matrix}
F & Y \\
0 & Q
\end{matrix}
\right),
\end{equation*}
where $E$ and $F$ are of size $n \times n$, $P$ and $Q$ are difference operators 
and $X$ and $Y$ have size $n \times 1$. 
First, let $EG=FH$ be a generator of the right ideal $E\cM_n(\cR) \cap F 
\cM_n(\cR)$ in $\cM_n(\cR)$,  $P \hat Q=Q \hat 
P$ be a generator of the right ideal $P \cR \cap Q \cR$ in $\cR$ and $K$ be a 
generator of the right ideal  
$E\cM_n(\cR)+ F \cM_n(\cR)$ in $\cM_n(\cR)$ (which is also called the greatest 
left common divisor of $E$ and $F$). We 
have by the induction hypothesis $\Ord K=\Ord F-\Ord G$. One can find a 
difference operator $R$ with $\Ord R \leq \Ord 
K$ and a vector difference operator $Z$ such that $KZ=(Y \hat P-X \hat Q)R$. 
Indeed, by Lemma \ref{lem1} one can assume that $K$ is a diagonal matrix 
$diag(K_0,...,K_{n})$. Let us call by 
$L_0,...,L_n$
the entries of the vector $Y \hat P-X \hat Q$. Then we can find for all 
$i=0,...n$ difference operators $M_i$ and $N_i$ 
such that $\Ord N_i \leq \Ord K_i$ and $K_iM_i=L_iN_i$. Let $R$ be a generator 
of the right ideal $N_0 \cR \cap ... \cap 
N_n \cR$. Then $\Ord R \leq \sum_{i=0}^n{\Ord N_i} \leq \sum_{i=0}^n{\Ord 
K_i}=\Ord K$ and there exists a vector $Z$ 
such that $KZ=LR$. Finally, by definition of $K$ there exist two matrix 
difference operator $V$ and $W$ such that 
$EV-FW=K$. Let
\begin{equation*}
C=
\left(
\begin{matrix}
H & WZ \\
0 & \hat P R
\end{matrix}
\right),
\hspace{2 mm}
D=
\left(
\begin{matrix}
G & VZ \\
0 & \hat Q R
\end{matrix}
\right).
\end{equation*}
Then $\Ord D \leq \Ord B$ and $AD=BC$.

The proof of the remaining parts of the statement are identical to the scalar 
case, see the proofs of Propositions \ref{Oreprop} and \ref{pro52}.
\end{proof}

The inequality (\ref{kerord}) is also true for a regular matrix difference 
operator $\cA\in\cM_n^\times(\cR)$ and we say that $\cA$ is a full kernel 
operator if $\mbox{Dim}_\cK \mbox{Ker}\, \cA=\Ord\,\cA$. Theorem \ref{thmrw}, 
Corollary \ref{cork} 
and Corollary \ref{inverse} from the previous section are also true for matrix 
rational operators.

\section{PreHamiltonian pairs and Nijenhuis operators}\label{Sec3}
Zhiber and Sokolov, in their study of Liouville integrable hyperbolic 
equations 
\cite{mr1845643}, have discovered  a family of 
special differential operators with the property 
that they define a new Lie bracket and are homomorphisms from the Lie algebra 
with the newly induced bracket to the 
original Lie algebra. These operators can be viewed as a generalization of 
Hamiltonian 
operators, although they are not necessarily anti--symmetric.  Inspired by the 
work of Zhiber and 
Sokolov, infinite sequences of such scalar differential operators of arbitrary 
order were constructed in 
\cite{mr1923781} using symbolic representation \cite{mr99g:35058, mnw08}.
Kiselev and van de Leur gave some examples of such matrix differential 
operators 
\cite{Kiselev2010} and investigated 
the geometric meaning of such operators. They named them preHamiltonian 
operators in \cite{kvdl1} and defined the 
compatibility of two such operators.
Recently, Carpentier renamed them as integrable pairs and investigated the 
interrelations between such pairs and 
Nijenhuis operators \cite{Carp2017}. In principle, many results for 
differential operators also work for difference 
operators since $\cR$ is a principal ideal domain. In this section, 
we develop further the theory of preHamiltonian operators  and extend it to 
the difference case. Similar to the previous section, we illustrate our results 
for the scalar case.

\begin{Def}
 A difference operator $A$ is called preHamiltonian if $\im A$ is a Lie 
subalgebra, i.e.,
 \begin{equation}\label{preH}
 [\im A,\ \im A]\subseteq \im A 
 \end{equation}
\end{Def}
By direct computation, it is easy to see that \cite{mr1923781} operator $A$ is 
preHamiltonian if and only if there 
exists a 2-form on $\cF$ denoted by $\omega_{A}$ such that
\begin{equation}\label{PreH1}
 A_*[A a](b)-A_*[A b](a)=A \omega_{A}(a,b)\quad \mbox{ for all }\quad a,b \in 
\cF .
\end{equation}
For a given $a \in \cF$, both $\omega_{A}(a, \bullet)$ and $\omega_{A}(\bullet, a)$ 
are in $\cR$, i.e. difference operators on $\cF$.

For a Hamiltonian operator $H$, the Jacobi identity is equivalent to
\begin{eqnarray}\label{loch}
 [H a, H b]=H\left( b_*[H a]+(Ha)_*^{\dagger}(b)-a_*[H 
b]+a_*^{\dagger}(Hb)\right),
\end{eqnarray}
for all $a, b \in \cF$, where $\dagger$ is the adjoint of the operator.
Clearly, Hamiltonian operators are preHamiltonian with $\omega_H(a, b)=(Ha)_*^{\dagger}(b)+a_*^{\dagger}(Hb).$
We are going to explore the relation between preHamiltonian pairs and
Hamiltonian pairs in the forthcoming paper \cite{CMW18}. Here we look at their relations with Nijenhuis operators.

Similar to Hamiltonian operators, in general, the
linear combination of two preHamiltonian operators is no longer preHamiltonian. 
This naturally leads to the following 
definition:
\begin{Def}\label{defpair}
 We say that two difference operators $A$ and $B$ form a preHamiltonian pair if 
$A+\lambda B$ is 
preHamiltonian for all constant $\lambda \in \cK$.
\end{Def}
A preHamiltonian pair $A$ and $B$ implies the existence of 2-forms $\omega_A$, 
$\omega_B$ and $\omega_{A+\lambda 
B}=\omega_A +\lambda \omega_B$. They satisfy
\begin{eqnarray}\label{pair}
A_*[B a](b) +B_*[A a](b)-A_*[B b](a)-B_*[A b](a)=A \omega_B(a,b)+B 
\omega_A(a,b) 
 \  \mbox{ for all }\  a,b \in 
\cF .
\end{eqnarray}
Gel'fand \& Dorfman \cite{GD79} and Fuchssteiner \& Fokas
\cite{mr82g:58039,mr84j:58046} discovered the relations between Hamiltonian 
pairs and Nijenhuis operators. They 
naturally generate
Nijenhuis operators. In what follows, we show that preHamiltonian pairs also 
give rise to Nijenhuis operators.
This also explains why we chose the terminology `preHamiltonian' 
instead of `integrable' for such operators.
These operators naturally appear in describing invariant evolutions of 
curvature 
flows \cite{Manbeffawang13}.

\begin{Def}
 A difference operator $R$ is Nijenhuis if 
 \begin{equation}\label{here}
 [R a, \ R b]-R[R a, b]-R[a, R b]+R^2[a,b]=0 \quad \mbox{ for all }\quad a,b 
\in 
\cF .
 \end{equation}
\end{Def}
Clearly, a Nijenhuis operator is also preHamiltonian with 
$$\omega_R (a,b)=(Rb)_*[a]-(R a)_*[b]-R[a,b] .$$ 

For a rational operator $R=A 
B^{-1}$, which is defined on $\im B$,
we define the Nijenhuis identity as
%\begin{equation}\label{Nijr}
%[A a, \ A b]-R[A a, B b]-R[B a, A b]+R^2[B a,B b]=0 \quad \mbox{ for all 
%}\quad 
%a,b \in \cF .
%\end{equation}
\begin{equation} \label{Nijrbis}
\begin{split}
&A_*[Aa]-[(Aa)_*,A]
+ AB^{-1}AB^{-1}(B_*[Ba]-[(Ba)_*,B]) \\
&=AB^{-1}(B_*[Aa]+A_*[Ba]-[(Aa)_*,B]-[(Ba)_*,A]) \quad \mbox{ for all }\quad 
a \in \cF ,
\end{split}
\end{equation}
where the bracket denotes the commutator of two difference operators.

\begin{The}\label{prop1}
If two difference operators $A$ and $B$ form a preHamiltonian pair, then $R=A 
B^{-1}$ is 
Nijenhuis.
\end{The}
\begin{proof}
Since $A$ and $B$ are preHamiltonian we can write for all $a \in \cF$
\begin{equation}
\begin{split}
A_*[Aa]-[(Aa)_*,A]&=A(\omega_A(a,\bullet)+(Aa)_*-a_*A);\\
B_*[Ba]-[(Ba)_*,B]&=B(\omega_B(a,\bullet)+(Ba)_*-a_*B).
\end{split}
\end{equation}
Hence, we see that, 
provided that $A$ and $B$ are preHamiltonians, \eqref{Nijrbis} is equivalent to
\begin{equation}
\begin{split}
&A B^{-1} \left(B\omega_A(a,\bullet)+A \omega_B(a,\bullet)-B_*[Aa]-A_*[Ba] \right.\\&\left.
+(Aa)_*B+(Ba)_*A-Aa_*B-Ba_*A\right)=0,
\end{split}
\end{equation}
where the expression inside the parentheses is nothing else than \eqref{pair}. Therefore, given two preHamiltonians 
difference operators $A$ and $B$, the ratio $AB^{-1}$ is Nijenhuis if and only 
if $A$ and $B$ form a preHamiltonian pair.
\end{proof}

Conversely, we have the following statement:
\begin{The}\label{prop2}
 Let $R$ be a Nijenhuis rational difference operator with minimal decomposition 
$AB^{-1}$ such that $B$ is preHamiltonian. Then $A$ and $B$ form a 
preHamiltonian pair. 
\end{The}
\begin{proof} 
Since $B$ is preHamiltonian, we have for all $a \in \cF$
\begin{equation}
B_*[Ba]-[(Ba)_*,B]=B(\omega_B(a,\bullet)+(Ba)_*-a_*B)
\end{equation}
Therefore, we can transform \eqref{Nijrbis} into the equivalent form
\begin{equation}
\begin{split}
&A_*[Aa]-(Aa)_*A+Aa_*A= \\
& =AB^{-1}(B_*[Aa]+A_*[Ba]+Ba_*A+Aa_*B-(Ba)_*A-(Aa)_*B-A \omega_B(a,\bullet).
\end{split}
\end{equation}
Let $CA=DB$ be the left least common multiple of the pair $A$ and $B$. It is 
also the right least common multiple of 
the pair $C$ and $D$ since $AB^{-1}$ is minimal. Therefore there exists a form 
$\omega_A$ such that 
\begin{equation}
\begin{split}
&A_*[Aa]-(Aa)_*A+Aa_*A=A \omega_A(a, \bullet); \\
& B_*[Aa]+A_*[Ba]+Ba_*A+Aa_*B-(Ba)_*A-(Aa)_*B-A \omega_B(a,\bullet)=B 
\omega_A(a,\bullet),
\end{split}
\end{equation}
which implies that $A$ and $B$ form a preHamiltonian pair.
\end{proof}

It is much easier to determine whether an operator is preHamiltonian than Hamiltonian. It can be systematically
done by any computer algebra software. Theorem 
\ref{prop1} provides an efficient 
method to check Nijenhuis property for rational operators, which is 
important in the theory of integrability.
\begin{Ex} The operators $A$ and $B$ defined in (\ref{volab}) form a 
preHamiltonian pair. Thus the recursion operator 
for the Volterra chain (\ref{vol}) is Nijenhuis.
\end{Ex}
\begin{proof} Let $C=A+\lambda B$. According to Definition \ref{defpair}, we 
check the existence of 2-form 
$\omega_C$ in (\ref{PreH1}). By direct computation, we have
\begin{eqnarray*}
&&\qquad C_*[C a](b)-C_*[C b](a)\\
&&\hspace{-0.6cm}= u \langle u_1 u_2 a_3 b_2+(u_1+u) u_1 a_2 
b_1+(u+u_{-1})u_{-1}a_{-1}b+u_{-2} 
u_{-1} a_{-2} b_{-1}+\lambda 
u_1 a_2 b_1+\lambda u_{-1}a_{-1} b\rangle_{\mathcal{P}_{a,b}}, 
\end{eqnarray*}
where ${\mathcal{P}_{a,b}} $ stands for anti-symmetrization with respect to 
$a_i$'s and $b_j$'s. We can now 
compute its preimage $\omega_C(a, b)$ by comparing its highest order either of 
$a$ or $b$ and we get
$$\omega_C(a, b)=u (a_1 b-a b_1)+u_{-1} ( a b_{-1} -a_{-1} b).$$
It follows from Theorem \ref{prop1} that the recursion operator $R=A B^{-1}$ 
for 
the Volterra chain (\ref{vol}) is 
Nijenhuis. 
\end{proof}

The previous two theorems provide the interrelations between preHamiltonian 
pairs 
and Nijenhuis operators. The following 
theorem (the analogue of its differential version which can be found in 
\cite{Carp2017}) gives another motivation for the definition of preHamiltonian 
pair: it is a necessary condition for a 
rational difference 
operator $R=AB^{-1}$ to `generate' an infinite commuting hierarchy.
\begin{The}
Let $R$ be a rational difference operator with minimal decomposition 
$R=AB^{-1}$. 
Suppose that there exist $(f^{(n)})_{n 
\geq 0} \in \cF$ spanning an infinite dimensional space over $\cK$ such that 
for 
all $n \geq 0$, 
$A(f^{(n)})=B(f^{(n+1)})$ and such that $[B(f^{(n)}), B(f^{(m)})]=0$ for all $n,m\geq 0$. 
Then $A$ and $B$ form a preHamiltonian 
pair.
\end{The}
\begin{proof}
 Since $[B(f^{(m)}),B(f^{(n)})]=0$ for all $m,n \geq 0$ by assumption, we have
\begin{equation}
(B_*[B(f^{(n+1)})]-(B(f^{(n+1)}))_*B)(f^{(m)})=B\left(-(f^{(m)})_*[B(f^{(n+1)})]\right) 
\hspace{2 mm} \forall \hspace{1 mm} m,n 
\geq 0.
\end{equation}
Similarly, replacing $B$ with $A$ we get for all $n,m \geq 0$ 
\begin{equation}
(A_*[A(f^{(n)})]-(A(f^{(n)}))_*A)(f^{(m)})=A(-(f^{(m)})_*[A(f^{(n)})]).
\end{equation}
Let $CA=DB$ be the left least common multiple of the pair $A$ and $B$. A non-zero 
difference operator has a finite 
dimensional kernel over $\cK$, therefore 
one must have for all $n \geq 0$ that
\begin{equation}
D(B_*[B(f^{(n+1)})]-(B(f^{(n+1)}))_*B)=C(A_*[A(f^{(n)})]-(A(f^{(n)}))_*A).
\end{equation}
By minimality of the fraction $AB^{-1}$, we deduce that for all $n \geq 0$ 
there 
exists a difference operator $P^{(n)}$ 
such that
\begin{equation} \label{auxeq}
\begin{split}
B_*[B(f^{(n+1)})]-(B(f^{(n+1)}))_*B&=BP^{(n)}, \\
A_*[A(f^{(n)})]-(A(f^{(n)}))_*A&=AP^{(n)}.
\end{split}
\end{equation}
For all $f \in \cF$ we can write $B(f)_*=B f_*+(D_B)_f$, where $(D_B)_f$ is defined by 
$(D_B)_f(g)=B_*[g](f)$ for all $g \in \cF$. $D_B$ is a 
bidifference operator, 
i.e., $(D_B)_f$ is a difference operator and its coefficients are difference 
operators applied to $f$. In other 
words $(D_B)_f=P_M(f)S^M+...P_N(f)S^N$ for all $f$, where $P_M,...P_N$ are 
difference operators. We can find a unique pair 
of bidifference operators
$Q$ and $R$ such that $\Ord R_f < \Ord B$ for all $f$ and 
\begin{equation} \label{auxeq2}
B_*[B(f)]-(D_B)_f B=BQ_f+R_f .
\end{equation}
From (\ref{auxeq}) we see that $R_{f^{(n)}}=0$ for all $n \geq 1$. This implies 
that 
$R=0$ since the $f^{(n)}$ span an 
infinite 
dimensional space over $\cK$. Therefore, for all $f,g$, we have
\begin{equation*}
\begin{split}
B_*[B(f)](g)-B_*[B(g)](f)
                =B Q_f(g)
  \end{split}              
\end{equation*}
implying that $B$ is preHamiltonian. 
Finally, since  for all constant $\lambda$, operator $$R+ \lambda=(A+ \lambda 
B)B^{-1}$$ satisfies the same hypothesis as 
$R$,
we conclude that $A+\lambda B$ is preHamiltonian. 
\end{proof}

\section{Basic definitions of differential-difference equations}\label{Sec4}
In this section we introduce some basic concepts for differential-difference 
equations relevant to the contents of this 
paper.
More details on the variational difference complex and Lie derivatives can be 
found in \cite{kp85,mwx2}.

Let $\bu =(u^1(n,t),\ldots ,u^N(n,t))$ be a vector function  
of a discrete variable $n\in \bbbz$ and time variable $t$, where $n$ and $t$ 
are 
``independent variables'' and $\bu$ will play
the role of a ``dependent'' variable in an evolutionary differential-difference 
system
\begin{equation}\label{evol}
 \bu_t=\f(\bu_p,\ldots,\bu_q),\qquad p\le q,\ \ p,q\in\bbbz.
\end{equation}
Equation (\ref{evol}) is an abbreviated form to encode the infinite 
sequence of ordinary differential systems of equations
\begin{equation*}
 \partial _t\bu(n,t)=\f(\bu(n+p,t),\ldots,\bu(n+q,t)),\qquad 
n\in\bbbz.
\end{equation*}

A vector function $\f$ is assumed 
to be a locally holomorphic function in its arguments. In the majority of cases 
it will be a rational or polynomial function which does not depend explicitly 
from the variables $n,t$. 
The corresponding vector field coincides with (\ref{Xf}). 
Thus there is a bijection between evolutionary derivations of 
$\cF$ and differential-difference systems with $\f\in\cF^N$.

\begin{Def} There are three equivalent definitions of  symmetry of 
an evolutionary equation.  We say that $\g\in\cF^N$ is a symmetry of  
(\ref{evol}) if
\begin{enumerate}
 \item $[\g,\f]=0.$
 \item  $\hat{\bu}_k=\bu_k+\epsilon \g_k$ satisfy equation (\ref{evol}) ${\rm 
mod}\,\epsilon^2$ whenever $\bu$ is a solution.
\item
Equation $\bu_\tau=\g$ is compatible with  (\ref{evol}).
\end{enumerate}
\end{Def}

Symmetries of an equation form a Lie subalgebra in $\Der \,\cF$. The existence 
of an infinite dimensional commutative Lie algebra of symmetries is a 
characteristic property of an integrable equation and it can be taken as a 
definition of integrability.

Often the symmetries of integrable equations can be generated by recursion 
operators
\cite{mr58:25341}. Roughly speaking, a recursion operator
is a linear operator $R: \cF^N\rightarrow \cF^N$ mapping a symmetry to a new 
symmetry.  For an evolutionary equation 
(\ref{evol}), it satisfies
\begin{equation}\label{reopev}
R_t=R_*[\f]=[ \f_*,\ R] .
\end{equation}
Recursion operators for nonlinear integrable equations are often Nijenhuis
operators.
Therefore, if the Nijenhuis operator $R$ is a recursion operator of 
(\ref{evol}), the operator $R$ is also a
recursion operator for each of the evolutionary equations in the hierarchy 
$\bu_t=R^k(\f)$, where $k=0,1,2,\ldots \ .$

Nijenhuis operators are closely related to Hamiltonian and symplectic operators.
The general framework 
in the context of difference variational complex and Lie derivatives can be 
found in \cite{kp85,mwx2}.
Here we recall the basic definitions related to Hamiltonian systems.

For any element $a\in \cF$, we define an equivalent class (or a functional) 
$\int\! a$
by saying that two elements $a,b\in\cF$ are equivalent if \(a-b\in
\mbox{Im}(\cS-1)\).  The space of functionals is denoted by $\cF'$.

For any functional $\int\!\!f\in \cF'$ (simply written $f\in \cF'$ without 
confusion), we define its difference variational derivative (Euler operator) 
denoted by
$\delta_{\bu} f \in \cF^N$ (here we identify the dual space with itself) as
$$\delta_{\bu} f=\left(\delta_{u^1} f, \cdots, \delta_{u^N} f \right)^{tr},\qquad 
\delta_{u^l} f=\sum_{i\in\bbbz} \cS^{-i}  \frac{\partial f}{\partial u_{i}^l}=
\frac{\partial }{\partial u^l}\left(\sum_{i\in\bbbz} \cS^{-i}  f \right).$$

\begin{Def}
An evolutionary equation (\ref{evol}) is said to be a Hamiltonian equation if 
there exists a Hamiltonian operator $H$
 and a Hamiltonian $\int g\in \cF'$ such that
$\bu_t=H \delta_{\bu} \int g. $
\end{Def}

This is the same to say that the evolutionary vector field $\f$ is a 
Hamiltonian 
vector field and thus the Hamiltonian operator
is invariant along it, that is, 
\begin{equation}\label{hameq}
 H_t=H_*[\f]=\f_* H+H \f_*^{\dagger} .
\end{equation}

Nijenhuis recursion operators for some integrable difference equations, e.g., 
the Narita-Itoh-Bogoyavlensky lattice \cite{wang12}, are no longer weakly 
nonlocal, 
but rational difference operators of the form $R=AB^{-1}$. The following 
statement tells us how operators $A$ and $B$ 
are related to a given equation.
\begin{The}\label{preP} If a rational difference operator $R$ with minimal 
decomposition 
$AB^{-1}$ is a recursion operator for 
equation (\ref{evol}), then there exists an difference operator $P$ such that 
the same 
relation
\begin{equation}\label{reab}
 Q_t=Q_*[\f]=\f_*Q+QP.
\end{equation}
is satisfied for both $Q=A$ and $Q=B$.
\end{The}
\begin{proof}
 Since $R=AB^{-1}$ is a recursion operator of (\ref{evol}), substituting it 
into 
(\ref{reopev}) we have
 $$
 R_t=A_*[\f]B^{-1}-A B^{-1} B_*[\f] B^{-1}= \f_* A B^{-1}-A B^{-1}\f_* ,
 $$
 that is,
 \begin{equation}\label{prfab}
 \left( A_*[\f]-\f_*A \right)=A B^{-1}\left( B_*[\f]-\f_*B\right). 
 \end{equation}
Let $CA=DB$ be the left least common multiple of the pair $A$ and $B$. It is 
also the right least common multiple of 
the pair $C$ and $D$. Moreover, we rewrite (\ref{prfab}) as 
$$
C \left( A_*[\f]-\f_*A \right)=D \left( B_*[\f]-\f_*B\right). 
$$
Therefore there exists a operator $P$ such that
\begin{eqnarray*}
&&  A_*[\f]-\f_*A=A P,\qquad  B_*[\f]-\f_*B=B P.
\end{eqnarray*}
Thus operators $A$ and $B$ satisfy the same relation (\ref{reab}).
\end{proof}
Comparing to (\ref{hameq}), for Hamiltonian operators, we have 
$P=\f_*^{\dagger}$. 
Conversely, it can be easy to show that
\begin{Pro} For an equation (\ref{evol}) if there exist two operators $A$ and 
$B$ satisfy (\ref{reab}), then 
$R=AB^{-1}$ is a recursion operator for the equation.
\end{Pro}
\begin{proof} By direct computation, we have
\begin{eqnarray*}
R_t=A_*[\f]B^{-1}-A B^{-1} B_*[\f] B^{-1}= (\f_*A+A P)B^{-1}-A B^{-1} (\f_*B+B 
P)B^{-1}=\f_*R-R\f_* 
\end{eqnarray*}
satisfying (\ref{reopev}). Thus $R=A B^{-1}$ is a recursion operator.
\end{proof}
This proposition has been used in \cite{Mikhailov2000} in constructing recursion 
operators for integrable noncommutative 
ODEs.
\begin{Ex}
For the operators $A$ and $B$ defined in (\ref{volab})
of the Volterra chain (\ref{vol}), the difference operator $P$ 
in Theorem \ref{preP} is $P=(1+\cS^{-1})u (1-\cS)$.
\end{Ex}

In what follows, we give the conditions for a rational recursion operator $R=AB^{-1}$ to
generate infinitely many local commuting symmetries. We first prove the following lemma:

\begin{Lem} \label{recpreham}
Assume that  $B$ is a preHamiltonian operator $R=AB^{-1}$ with minimal decomposition is a recursion operator 
for $\bu_t=B(\g)$, where $\g\in \cF^N$. Then $[B(\g),A(\g)]=0$.
\\ In particular, if there exists $\h \in \cF^N$ such that
$R$ is a recursion operator for $\bu_t=B(\h)$ and $[B(\g),A(\h)]=0$, then $[A(\g),B(\h)]=0$.
\end{Lem}
\begin{proof}
We know that $B$ is preHamiltonian. So for any $\a\in \cF^N$, we have
\begin{eqnarray}\label{preB}
B_*[B\a]-(B\a)_*B &=B(\omega_B(\a,\bullet)-\a_*B). 
\end{eqnarray}
From Theorem \ref{preP}, it follows, when $\a=\g$ or $\a=\h$, that
\begin{eqnarray}\label{reAB}
A_*[B\a]-(B\a)_*A &=A(\omega_B(\a,\bullet)-\a_*B).
\end{eqnarray}
Using (\ref{reAB}) for $\a=\g$, we get
\begin{eqnarray*}
[B(\g),A(\g)]= A_*[B(\g)](\g)+A{\g}_*[B(\g)]-(B\g)_*[A \g]=A(\omega_B(\g,\g))=0.
\end{eqnarray*}
If there exists $\h \in \cF^N$ such that
$R$ is a recursion operator for $\bu_t=B(\h)$ then from the former we deduce that
\begin{equation}
[B(\g+\h),A(\g+\h)]=[B(\g),A(\g)]=[B(\h),A(\h)]=0.
\end{equation}
Hence $[B(\h),A(\g)]=-[B(\g),A(\h)]$.
\end{proof}
\iffalse
Combining the definition of preHamiltonian pair with Theorem \ref{preP}, we see that for $h=f,g$ we have
\begin{equation} \label{equationaux}
\begin{split}
B_*[Bh]-(Bh)_*B &=B(\omega_B(h,\bullet)-h_*B)\\
A_*[Bh]-(Bh)_*A &=A(\omega_B(h,\bullet)-h_*B)
\end{split}
\end{equation}
If we take $h=f$ in \ref{equationaux} and apply both equations to $g$ we get
\begin{equation}
\begin{split}
 [B(f),B(g)]&=B(\omega_B(f,g)-f_*Bg+g_*Bf)\\
[B(f),A(g)] &=A(\omega_B(f,g)-f_*Bg+g_*Bf)
\end{split}
\end{equation}
which implies by the hypothesis (recall that $A$ and $B$ are right coprime) that $\omega_B(f,g)-f_*Bg+g_*Bf=0$. Finally, 
take $h=g$ in \ref{equationaux} and applying the second equation to $f$ we get $[A(f),B(g)]=0$.
\fi

\begin{Pro} \label{com}
Assume that $A$ and $B$ form a preHamiltonian pair and $R=AB^{-1}$ is a recursion operator 
for $\bu_t=B(\g^{(0)})$, where $\g^{(0)}\in \cF^N$. If there exists $\g^{(n)}\in\cF^N$ such that
 $A(\g^{(n)})=B(\g^{(n+1)})$ for all $n \geq 0$, then $[B(\g^{(n)}),B(\g^{(m)})]=0$ for all $n,m \geq 0$.
\end{Pro}
\begin{proof} We can assume that $AB^{-1}$ is a minimal decomposition of $R$. Indeed, if not we write $A=A_0C$ and $B=B_0C$ where $R=A_0B_0^{-1}$ is minimal and replace $\g^{(n)}$ by $C(\g^{(n)})$.  By Theorem \ref{prop1}, we know that $R$ is Nijenhuis and thus it is a recursion operator for all $B(\g^{(n)})$, $n\geq 0$. 
We proceed the proof by induction on $|n-m|$. If $n=m$ there is nothing to prove. If $|n-m|=1$, we deduce $[B(\g^{(n)}),B(\g^{(n+1)})]=0$ as a direct application of Lemma \ref{recpreham} 
since $B(\g^{(n+1)})=A(\g^{(n)})$ for all $n \geq 0$. Suppose that $[B(\g^{(n)}),B(\g^{(m)})]=0$ for all $n,m \geq 0$ 
such that $|n-m| \leq N$, which implies
$[B(\g^{(n+N)},B(\g^{(n+1)})]=0$. Hence by Lemma \ref{recpreham}, 
we have $[B(\g^{(n+N+1)},B(\g^{(n)})]=0$.
\end{proof}

\section{Application: a new rational recursion operator}\label{Sec5}
In this section, we construct a recursion operator of system 
\begin{equation}\label{bv}
 u_t=u^2(u_2 u_1-u_{-1} u_{-2})-u (u_1-u_{-1}):=f
\end{equation}
from its Lax
representation and show that it is Nijenhuis and generates local commuting symmetries.
In general, it is not easy to construct a recursion operator for a
given integrable equation although the explicit formula is given. The 
difficulty 
lies in how to determine the starting terms of $R$, i.e.,
the order of the operator, and how to construct its nonlocal terms. Many papers 
are
devoted to this subject, see \cite{mr1974732, mr88g:58080, hssw05}. If the Lax
representation of the equation is known, there is an amazingly simple approach 
to
construct a recursion operator proposed in \cite{mr2001h:37146}. The idea in
\cite{mr2001h:37146} can be developed for the Lax pairs that are invariant 
under 
the
reduction groups, which applied for both differential and difference equations 
\cite{wang09, wang12}.

Equation (\ref{bv}) first appeared in \cite{adler2}, where the authors 
presented 
its scalar Lax representation. We rewrite it in the matrix form as follows:
\begin{eqnarray}
&& L=\cS-\bU(\lambda)=\cS-\lambda \bU^{(1)}-\bU^{(0)}=\cS- \left(\begin{matrix}
                      0 &1 &0\\0 &0 &1\\\lambda & -\frac{1}{u} & 
\frac{\lambda}{u}
                     \end{matrix}\right) \label{A}\\
&& M=D_t-\bV(\lambda)=D_t+ \left(\begin{matrix}
                      -\frac{1}{\lambda^2}+u_{-1} &\frac{1}{\lambda}(1-u_{-1} 
u_{-2}) &u_{-1}u_{-2}-\frac{u_{-1}}{\lambda^2}\\\lambda u 
u_{-1}-\frac{u}{\lambda} &u-u_{-1} &\lambda u_{-1}-\frac{u u_{-1}}{\lambda}\\
                      \lambda^2 u- u u_1 & \lambda (u u_1-1) & \lambda^2-u
                     \end{matrix}\right)\label{B}   ,                  
\end{eqnarray}
where $\lambda$ is a spectral parameter.
The commutativity of the above operators leads to the zero curvature condition
\begin{equation}\label{lax}
\bU(\lambda)_t=\cS \bV(\lambda)\cS^{-1} \bU(\lambda)-\bU(\lambda) \bV(\lambda)
\end{equation}
and subsequently it leads to system (\ref{bv}). System (\ref{bv}) defines a 
derivation $X_f\in \cA_1$ of $\cR$ with $f=u^2(u_2 u_1-u_{-1} u_{-2})-u 
(u_1-u_{-1})$.
The representation (\ref{A}), (\ref{B}) is invariant with respect to 
transformations:
\begin{equation}
 \label{inv1}
 \cS \bV(\lambda)\cS^{-1}=-J \cT \bV(\lambda^{-1}) \cT J,\qquad 
\bU^{-1}(\lambda)= J \cT \bU(\lambda^{-1}) \cT J
\end{equation}
and
\begin{equation}
 \label{inv2}
 \bV(\lambda)=H  \bV(-\lambda) H,\qquad 
\bU(\lambda)= -H\bU(-\lambda)H,
\end{equation}
where
\[
 J=\left(\begin{array}{rrr}
          0&0&1\\0&1&0\\1&0&0
         \end{array}\right),\qquad 
H=\left(\begin{array}{rrr}
          1&0&0\\0&-1&0\\0&0&1
         \end{array}\right)\, .
\]
Transformation (\ref{inv1}) reflects the symmetry $\cT(f)=-f$ of the equation 
(\ref{bv}).

For given matrix $\bU$, we can build up a hierarchy of nonlinear systems by 
choosing
different matrices $\bV$ with the degree of $\lambda$ from $-2l$ to $2l$.
The idea to construct a recursion operator directly from a Lax representation 
is 
to
relate the different operators $\bV$ using ansatz $${\bar 
\bV}=(\lambda^{-2}+\lambda^2) \bV
+\bW$$ and then to find the relation between two flows corresponding to $\bar 
\bV$
and $\bV$. The multiplier $\lambda^{-2}+\lambda^2$ is the automorphic function 
of the group generating by transformations $\lambda\mapsto \lambda^{-1}$ and 
$\lambda \mapsto-\lambda$. Here  $\bW$ is the reminder and we assume that it 
the same symmetry as $\bV$:
\begin{eqnarray}\label{ansat}
\bW=\sum_{j=-2}^{2} \la^{j} \bW^{(j)} ,
\end{eqnarray}
where $\bW^{(j)}$ are $3\times 3$ matrices of the following form (invariant under (\ref{inv2}))
\begin{eqnarray*}
 && \bW^{(2)}=\left(\begin{matrix}
             a&0&b\\0&c&0\\d&0&e
            \end{matrix}\right); \quad \bW^{(0)}=\left(\begin{matrix}
             a_0&0&b_0\\0&c_0&0\\d_0&0&e_0
            \end{matrix}\right);\quad \bW^{(-2)}=\left(\begin{matrix}
             a_{-}&0&b_{-}\\0&c_{-}&0\\d_{-}&0&e_{-}
            \end{matrix}\right);\\
&& \bW^{(1)}=\left(\begin{matrix}
             0&r&0\\s&0&p\\0&q&0
            \end{matrix}\right); \quad
            \bW^{(-1)}=\left(\begin{matrix}
             0&r_{-}&0\\s_{-}&0&p_{-}\\0&q_{-}&0
            \end{matrix}\right)
\end{eqnarray*}
and following from $\bW$ is invariant under (\ref{inv1}), they satisfy
\begin{eqnarray*}
 \cS (\bW^{(2)})=-J \cT(\bW^{(-2)}) J; \quad \cS (\bW^{(1)})=-J \cT(\bW^{(-1)}) J; \quad \cS (\bW^{(0)})=-J \cT(\bW^{(0)}) J.
\end{eqnarray*}
The zero curvature condition leads to
\begin{eqnarray}\label{Laxt}
\begin{array}{l}
\bU_{\tau}=(\frac{1}{\lambda^2}+\lambda^2)\bU_t+\cS(\bW) \bU-\bU \bW .
\end{array}
\end{eqnarray}
Substituting the ansatz (\ref{ansat}) into it and collecting the coefficient
of powers of $\la$, we obtain six matrix equations for $\bW^{(j)}, j=-2, 
\cdots, 
2$. For example, the equation corresponding to linear terms of $\la$ is
\begin{eqnarray}\label{relt}
\bU^{(1)}_{\tau}=\cS(\bW^{(1)}) \bU^{(0)}+\cS(\bW^{(0)}) \bU^{(1)}- \bU^{(1)} 
\bW^{(0)}-\bU^{(0)} \bW^{(1)}.
\end{eqnarray}
Through them we are able to determine the entries of matrices $\bW^{(j)}$ and 
we 
finally get
\begin{eqnarray}
&&c_{-}=(\cS^2-1)^{-1}\frac{u_t}{u};\nonumber\\
&&b_0=\cS^{-1}u (\cS u-u \cS^{-1})^{-1} \left(u(\cS-\cS^{-2})u 
(\cS^2+\cS+1)+\cS^2-\cS \right) c_{-};\nonumber\\
%&&r_{-}=\cS^{-1} u (u\cS-\cS^{-1} u)^{-1} \left(u(\cS^{-2}u-\cS u+ \frac{1}{u} 
%\cS)(\cS^2+\cS+1)-\frac{1}{u_{-1}}(2+\cS) \right)c_{-};\\
&&e_0=(\cS^2+\cS+1)^{-1} \left(-(\cS\frac{1}{u} \cS+\frac{1}{u} 
\cS+\cS\frac{1}{u}) \cS b_0+(\cS+1)\frac{1}{u}(\cS^2-\cS)c_{-}\right)  
;\nonumber\\
&&u_{\tau}=u^2(\cS^3-1)b_0-u (u \cS-\cS^{-1} u) (\cS^2+\cS+1) c_{-}+u 
(\cS-1)e_0.\label{taut}
\end{eqnarray}
Note that
$$
\cS\frac{1}{u} \cS u+\frac{1}{u} \cS u+\cS=(\cS+1) \frac{1}{u} (\cS u-u 
\cS^{-1})+\cS^{-1}+1+\cS .
$$
We simplify the above expression of $e_0$. It becomes
$$
e_0= (\cS^{-2}-1)u (\cS^2+\cS+1)  c_{-}
-\cS^{-1} (\cS u-u \cS^{-1})^{-1} \left(u(\cS-\cS^{-2})u 
(\cS^2+\cS+1)+\cS^2-\cS 
\right) c_{-}
$$

Substituting $c_{-}, b_0$ and $e_0$ into (\ref{taut}), we obtain the relation 
between two symmetry flows $u_t$ and $u_{\tau}$. Thus we obtain the following 
statement:
\begin{Pro}\label{pro51} A recursion operator for equation (\ref{bv}) is
\begin{eqnarray}
&&R=u\left(u (\cS^2\!-\!\cS^{-1}\!) u +\cS^{-1}\!\!\!-\!1\right)(\cS 
u\!-\!u\cS^{-1}\!)^{-1}\left( u (\cS\!-\!\cS^{-2}\!)u 
(\cS^2+\cS+1)+\cS^2\!-\!\cS \right)(\cS^2\!-\!1)^{-1}\! \frac{1}{u}\nonumber\\
&&\quad +u(2 \cS^{-1} u-\cS^{-2} u-\cS u+u-u\cS) 
(\cS^2+\cS+1)(\cS^2-1)^{-1}\frac{1}{u}. \label{readler}
\end{eqnarray}
\end{Pro}
We represent $R$ as
\begin{eqnarray*}
R=R^{(3)}+R^{(1)}+R^{(-1)},
\end{eqnarray*}
where 
\begin{eqnarray*}
&& R^{(3)}=u^2 (\cS^3-1)\cS^{-1}u (\cS u-u\cS^{-1})^{-1} u (\cS-\cS^{-2})u 
(\cS^2+\cS+1) (\cS^2-1)^{-1} \frac{1}{u};\\
%&& \Re^{(1)}=u(2 \cS^{-1} u-\cS^{-2} u-\cS u+u-u\cS) 
%(\cS^2+\cS+1)(\cS^2-1)^{-1}\frac{1}{u}\\
%&&\qquad +u^2(\cS^2-\cS^{-1}) u (\cS u-u \cS^{-1})^{-1} \cS (\cS+1)^{-1} 
%\frac{1}{u}\\
%&&\qquad -u (1-\cS^{-1}) (\cS u-u \cS^{-1})^{-1} u(\cS-\cS^{-2})u 
%(\cS^2+\cS+1)(\cS^2-1)^{-1}\frac{1}{u};\\
&&R^{(-1)}=u (\cS^{-1}-1) (\cS u-u \cS^{-1})^{-1} (\cS^{-1}+1)^{-1} \frac{1}{u}.
\end{eqnarray*}
Note that $R^{(3)}$ is a recursion operator for $u_t=u^2(u_1 u_2-u_{-1} 
u_{-2})$ 
\cite{wang12} and that $R^{(-1)}$ is the 
inverse recursion operator 
for the Volterra chain $u_t=u (u_1-u_{-1})$ \cite{kmw13}.

The recursion operator (\ref{readler}) is not weakly nonlocal. We now rewrite 
it 
as a rational difference operator. It is convenient to first write $R$ as
\begin{eqnarray}\label{readler1}
 R=\left(Q\Delta^{-1}C+P\right)(\cS^2-1)^{-1} \frac{1}{u},
\end{eqnarray}
where
\begin{eqnarray}
&& Q=u\left(uu_1-1+(1-uu_{-1})S^{-1}\right); \ \ \Delta=\cS u -u \cS^{-1}; 
\label{qq}\\
&& C=w_2\cS^2-w_{-1} \cS, \quad w=1-u_{-1} u_1;\label{cc}\\
&& P=\left(u^2\cS u (\cS-\cS^{-2})u+Q u_{-1}+ u(2 \cS^{-1} u-\cS^{-2} u-\cS 
u+u-u\cS)\right)(\cS^2+\cS+1)\nonumber\\
&&\quad +u^2 \cS (\cS^2-\cS)+Q(u_{-1} \cS^{-1}+u_2 \cS)=\tilde{P}(S+1)+p; 
\label{pp} 
\end{eqnarray}
where $\tilde P$ is a difference operator and $p=u(u_2+2u_1-2u_{-1}-u_{-2})$.

\begin{Lem}\label{lem51}
The recursion operator $R$ given by (\ref{readler1}) can be factorized as 
$R=AB^{-1}$ with 
\begin{equation}\label{Bb}
B=u(\cS-\cS^{-1})(\cS \alpha + \beta +\cS^{-1}\gamma ),
\end{equation}
where
\begin{eqnarray*}
&& \alpha=u_{-1}uw_{-1}w-u_{-1}u_1w_1w_2;\quad 
\beta=u^2w^2-u_{-1}u_1w_{-2}w_2;\\
&&\gamma=u_{1}uw_{1}w-u_{-1}u_1w_{-1}w_{-2}; \quad w=1-u_{-1} u_1.
\end{eqnarray*} 
and
\begin{equation}\label{Aa}
A=Q (\frac{1}{u}w_1 \alpha_1 \cS+\frac{1}{u_1} w \gamma) +P\cS^{-1} (S \alpha + 
\beta + S^{-1}\gamma ).
\end{equation}
\iffalse
\begin{equation}
A=A_{\alpha} \alpha+A_{\beta} \beta+ A_{\gamma} \gamma+E,
\end{equation}
where
\begin{eqnarray*}
&& 
A_{\alpha}=u^2u_1u_2S^4-uu_1w_1S^3-(uu_1w_1+uu_{-1})S^2+(uu_{-1}w_{-1}
-uu_1w_1-w_{-1})S+\frac{u}{u_{-1}}+uu_{-1}w_{-1};\\
&& 
A_{\beta}=u^2u_1u_2S^3-(uu_1w_1+u^2w_3)S^2-(uu_{-1}+uu_1w_1)S+uu_{-1}w_{-1}
-uu_1w_1+uu_{-2}-uu_2;\\
&&A_{\gamma}=\frac{u^2u_2}{u_3} S^2-(u^2w_3+uu_1w_1)S-\frac{u}{u_1}-uu_1w_1; \\
&&E=c^{-1}S^{-1}+c^{-2}S^{-2}+c^{-3}S^{-3}+c^{-4}S^{-4}, \text{ defined by } 
\sigma A \sigma=-A.
\end{eqnarray*}
\fi
\end{Lem} 
\begin{proof}
To find $A$ and $B$ for (\ref{readler1}) we need to rewrite $\Delta^{-1}C$ as a 
right fraction. It turns out that 
$$C \cS^{-1}(\cS \alpha + \beta + \cS^{-1}\gamma)= \Delta (\frac{1}{u}w_1 
\alpha_1 \cS+\frac{1}{u_1} w \gamma),$$ 
from which we can find that $\alpha, \beta$ and $\gamma$ as stated is a 
solution. Then $A=R B$ by definition as given in the statement.
\end{proof}
The authors in \cite{zhang02} showed that the recursion operators derived from 
certain Lax representations under 
certain boundary conditions are Nijenhuis once every step is uniquely 
determined.  Here we prove the
Nijenhuis property using the results in Section \ref{Sec3}.
\begin{The}\label{AdlerN}
The operators $A$ and $B$ defined by (\ref{Aa}) and (\ref{Bb}) are compatible 
preHamiltonian operators. In particular,
the recursion operator $R$ for equation (\ref{bv}) given by (\ref{readler}) is 
Nijenhuis.
\end{The}
\begin{proof}
We know from Lemma \ref{lem51} that $R=AB^{-1}$.  To prove that it is Nijenhuis, we 
only need to show operators $A$ 
and $B$ form a preHamiltonian pair following from Theorem \ref{prop1}. 

Let 
$I=A+\lambda B$. For any $a, b\in\cF$ and 
constant $\lambda$, we use computer algebra package Maple to compute
$e^{(0)}=I_*[Ia](b)-I_*[Ib](a)$, which is linear in $a$ and its shifts. We take 
the coefficient of the highest 
order term $a_k$ (here $k=11$) in $e^{(0)}$ and denote it by $v^{(0)}$. Notice that the 
highest order term in $I$ is $u^2 u_1 u_2 
\cS^4 \alpha$. We set $\omega^{(0)}=\frac{1}{\alpha}\cS^{-4}(\frac{v^{(0)} 
a_k}{u^2 u_1 u_2 })$. We then compute $ 
e^{(1)}=e^{(0)}-I(\omega^{(0)})$ and repeat the procedure. Finally we get 
$e^{(11)}=0$ after $n=11$ steps implying $I$ is 
preHamiltonian.
\end{proof}

Since the operator $R$ is not weakly nonlocal, the results on the locality of 
symmetries generated by $R$ in 
\cite{wang09} are no longer valid. In the rest of this section, we are going to 
show that $R$ generates 
infinitely many commuting symmetries of (\ref{bv})  starting from the equation 
itself. 

\begin{Pro}\label{propr}
Let $h$ be a difference polynomial such that $R$ is a recursion operator for 
$u_t=h$. Then 
$h$ lies in the image of $B$. More 
precisely
$h=B(x)$ for some $x \in \cF$ and $A(x)$ is a difference polynomial. Moreover, 
$R$
is a recursion operator for $u_t=A(x)$.
\end{Pro} 

We will break the proof of this proposition in two parts using \eqref{readler1}.
First we will prove that $h=u(g_2-g)$ for some difference polynomial $g$. Second 
we will show that $C(g)=\Delta(k)$ for some difference polynomial $k$. We begin with proving a few lemmas.
To improve the readability, we put them in Appendix B. We now write the proof for Proposition \ref{propr} using these lemmas.

\begin{proof}
By Lemma \ref{lem3}, we know that $h=u(g_2-g)$ for some difference polynomial 
$h$. By Lemma \ref{lem4} and Lemma 
\ref{lem7},
for some constant $\lambda \in \cK$  we get that
\begin{equation*}
C(g)_*u(S^2-1)-C(g)(S^2-S) \equiv \lambda C.
\end{equation*}
Since $g$ is a difference polynomial, the constant term in 
$C(g)_*u(S^2-1)-C(g)(S^2-S)-\lambda C$ is $\lambda 
(S-S^2)$. This constant
term must be divisible on the left by $\Delta$, which implies $\lambda=0$. 
Moreover, we can divide the congruence 
relation by $(S^2-S)$ on the
right since $\Delta$ has a trivial kernel:
\begin{equation*}
C(g)_*u(1+S^{-1}) \equiv C(g).
 \end{equation*}
  After applying Lemma \ref{lem8} we deduce that $C(g)=\Delta(k)$ for 
some difference polynomial $k$.

Let 
$M$ be a generator of the right ideal $C \cR \cap \Delta \cR$ in $\cR$. This 
means that
$M=CE=\Delta D$ for some pair of right coprime difference operators $D$ and $E$. 
By Lemma \ref{lem9},
there exists  $x \in \cF$ such that $g=E(x)$ and $k=D(x)$. Since $B=u(S^2-1)E$, 
we conclude that $h=B(x)$. Finally, 
$A=QD+PE$, 
hence $A(x)=P(g)+Q(k)$ is a difference polynomial. $R$ is a recursion operator for $u_t=A(x)$ 
since $R$ is Nijenhuis following from 
Theorem \ref{AdlerN}.
\end{proof}

\iffalse
\begin{Lem}
$\Re$ is recursion for the difference polynomials $F_{2n}$ for all $n \geq 1$.
\end{Lem}
\begin{proof}
For a polynomial $F$ we denote by $\mathcal{L}_F$ the following derivation of 
the ring of pseudo--difference operators 
$L$:
\begin{equation*}
\mathcal{L}_F(L)=X_F(L)-[D_F,L].
\end{equation*}
We know that $\Re$ is recursion for $F_2$ which means that 
$\mathcal{L}_{F_2}(\Re)=0$. Let $G=F_{2k}$ for some $k \geq 
1$. Assume that
$\mathcal{L}_{G}(\Re) \neq 0$ and let $bS^d$ be its leading term as a Laurent 
series in $S^{-1}$. Recall that the 
leading term of $\Re$ is $u^2u_1S^2=aS^2$.
Therefore, for all integer $m \geq 1$, the leading term of 
$\mathcal{L}_{G}(\Re^m)$ is $c^mS^{d+2(m-1)}$, satisfying the 
recurrence relation
$c^{m+1}=c^m a_{d+2(m-1)}+a...a_{2(m-1)}b_{2m}$. In particular, infinitely many 
$c^m$'s are non-zero, which means that 
the sequence of degrees
of the $\mathcal{L}_{G}(\Re^m)$ is unbounded. 
\\
On the other hand, from $[G,F_{2n}]=0$ we know that 
$d(\mathcal{L}_G(D_{F_{2n}}))=d(X_{F_{2n}}(D_G)) \leq 2k$ for all 
$n$. Moreover, for all $m \geq 1$ there exist constants $\lambda_1,..., 
\lambda_m$
such that the degree of $R^m-\lambda_m D_{F_{2m}}... -\lambda_1 D_{F_2}$ is at 
most zero. Therefore the degree of 
$\mathcal{L}_{G}(\Re^m)$ should be bounded by
$2k$. This is a contradiction, hence $\mathcal{L}_{G}(\Re) = 0$, i.e. $\Re$ is 
recursion for $G$.
\end{proof}
\fi
\begin{The}
 There exists a sequence $g^{(2)},g^{(4)},g^{(6)},...$ in $\cF$ such that 
\begin{enumerate} 
\item[(1)] $u^2(u_1u_2-u_{-1}u_{-2})-u(u_1-u_{-1})=B(g^{(2)})$; 
\item[(2)] $A(g^{(2n)})=B(g^{(2n+2)}) \text{  for all }n \geq 1$; 
\item[(3)] $B(g^{(2n)})$ is a difference polynomial for all $n \geq 1$; 
\item[(4)] $ [B(g^{(2n)}),B(g^{(2m)})]=0$ for all $n,m \geq 1$;
\item[(5)] The order of $B(g^{(2n)})$ is $(-2n,2n)$;
\item[(6)]  $R$ is a recursion operator for all the $u_t=B(g^{(2n)})$.
\end{enumerate}
%This hierarchy is complete.
Finally, let $V=\text{Span}_{\cK} \{ B(g^{(2n)}) | n \geq 1 \}$. If $f \in \cF$ commutes with some element 
$h \in V$, then $f \in V$.
\end{The}
\begin{proof}
We already know that $R$ is a recursion operator for \eqref{bv}, hence by Proposition 
\ref{propr} there exists $g^{(2)} \in 
\cF$ such 
that $(1)$ is satisfied and $A(g^{(2)})$ is a difference polynomial. Since $R$ is
Nijenhuis (following from 
Theorem \ref{AdlerN}) it must be a recursion operator for $u_t=A(g^{(2)})$ as well. Using Proposition 
\ref{propr} a second time 
we find 
$g^{(4)}\in \cF$ such that $B(g^{(4)})=A(g^{(2)})$ and $A(g^{(4)})$ is a difference 
polynomial. Iterating this argument we 
prove the statements $(2)$, $(3)$ and $(6)$. Statement $(5)$ is obvious and statement $(4)$ follows from Proposition 
\ref{com} and Theorem \ref{AdlerN}.
Finally, if $f \in \cF$ commutes with $h \in V$, let us 
sketch the proof of how to show that $f \in V$. If $(M,N)$ is the order of $f$ 
and $N>0$ it is not hard to prove from the equation
\begin{equation} \label {eq6}
X_h(f_*)=[h_*,f_*]+X_f(h_*).
\end{equation}
Note that the leading term of $f_*$ is up to multiplication by a constant the 
leading 
term of $B(g^{(2k)})_*$ for some $k \geq 2$. Similarly, if $M<0$, one sees that 
the negative leading term of $f_*$ is up to multiplication by a constant the 
negative leading term of $B(g^{(2l)})_*$ for some $l \geq 2$. We conclude by 
induction on the total order of $f$, after checking that the only $f$ commuting 
with an element of $V$ and which depend either on $u,...,u_N$ or on 
$u_{-N},...,0$ for $N \geq 0$ is $f=0$.
\end{proof}

\begin{Rem}
Note that $g^{(2)}=\frac{u_{-1}uu_1w_{-1}ww_1}{\alpha \gamma}$ is in $\cF$ but is not a difference polynomial.
\end{Rem}

\begin{Rem}
Let $\cT$ be the automorphism of $\mathrm{K}$ defined in section \ref{Sec2}. Then we have
 $\cT A \cT=-A$ and  $\cT B \cT=-B.$ This implies 
that $\cT (B(g^{(2n)}))=-B(g^{(2n)})$  and $\cT(g^{(2n)})=g^{(2n)}$ for 
all $n \geq 1$.
\end{Rem}

\section{Application: inverse Nijenhuis recursion operators}\label{Sec6}
In \cite{kmw13}, the authors listed integrable differential-difference equations 
with their algebraic properties. For some systems, they presented both recursion 
operators and their inverse in weakly nonlocal
form.  In this section, we'll explain the (non)existence of weakly nonlocal 
inverse recursion 
operators and how to work out the nonlocal terms based on Theorem \ref{thmrw} and its corollaries in 
Section \ref{sec23} using examples in 
\cite{kmw13}. 

We select four examples: in section \ref{toda}, we show the nonexistence of 
weakly nonlocal inverse 
recursion operator for the Toda lattice; in section \ref{secrtoda}, we show the 
existence of weakly nonlocal inverse 
recursion operator with only one nonlocal term for a relativistic Toda system; 
in section \ref{AL}, we deal with a 
recursion operator with two nonlocal terms; for our last example, we demonstrate 
that the inverse operator $R$ itself 
is not weakly nonlocal, but that of $R-{\rm id}$ is!
\subsection{The Toda Lattice}\label{toda}
The Toda equation \cite{toda67} is given by
\begin{eqnarray*}
q_{tt}=\exp(q_{1}-q)-\exp(q-q_{-1})
\end{eqnarray*}
In the Manakov-Flaschka coordinates \cite{flaschka1, manakov} defined by 
$u=\exp(q_1-q),\  v=q_{t}$, it can be rewritten as two-component evolution 
system:
\begin{eqnarray}\label{todaeq}
\left\{ {\begin{matrix} u_t= u(v_{1}-v)\\
   v_t=  u-u_{-1}  \end{matrix} } \right. ,
 \end{eqnarray}
which admits two compatible Hamiltonian local structures
\begin{equation*}
H_1=\left(\begin{matrix}
             0&u(\cS-1)\\(1-\cS^{-1})u&0
            \end{matrix}\right), \hspace{2 mm}
H_2=\left(\begin{matrix}
             u(\cS-\cS^{-1})u&u(\cS-1)v\\v(1-\cS^{-1})u&u\cS-\cS^{-1}u
            \end{matrix}\right)
\end{equation*}
It is clear to see that ${\rm Ord} H_1=2$ and that the kernel of $H_1$ is 
spanned by $\left(\begin{matrix} \frac{1}{u} \\ 0 \end{matrix} \right)$ and 
$\left(\begin{matrix} 0 \\ 1 \end{matrix} \right)$. 
One can check that the kernel of $H_2$ is spanned by 
$\left(\begin{matrix} \frac{1}{u} \\ 0 \end{matrix} \right)$. In other words, 
$H_1$ and $H_2$ have a common right divisor 
$H$ of the total order being $1$ and can be written as $BH$ and $AH$, where ${\rm Ord} B=1$
and ${\rm Ord} A=3$, that is,
\begin{eqnarray*}
&&H_1=BH=\left(\begin{matrix}
           0& u (1-\cS)\\1&0
            \end{matrix}\right)\left(\begin{matrix}
            (1-\cS^{-1})u&0\\0&1
            \end{matrix}\right);\\
&&H_2=AH=\left(\begin{matrix}
            u(\cS+1)&u(\cS-1)v\\v&u\cS-\cS^{-1}u
            \end{matrix}\right)\left(\begin{matrix}
            (1-\cS^{-1})u&0\\0&1
            \end{matrix}\right) .
\end{eqnarray*}
Thus $B$ has full kernel and $A$ has trivial kernel. Thus the recursion operator $$R=H_2 H_1^{-1}=AB^{-1}$$ is 
weakly non-local but $BA^{-1}$ is not. Indeed,
\begin{eqnarray*}
&& R=\left(\begin{matrix}  v_1& u\cS+u  \\
  1+\cS^{-1}&v \end{matrix}\right)+\left(\begin{matrix} u(v_1-v) \\u-u_{-1} \end{matrix}\right) (\cS-1)^{-1}
 \left(\begin{matrix} \frac{1}{u} & 0 \end{matrix}\right)
\end{eqnarray*}

\subsection{A Relativistic Toda system }\label{secrtoda}
The relativistic Toda system \cite{ruijsenaars1} is given by
 \begin{eqnarray*}
q_{tt}=q_{t}q_{-1t}\frac{\exp(q_{-1}-q)}{1+\exp(q_{-1}-q)}-q_{t}q_{1t}\frac{
\exp(q-q_{1})}{1+\exp(q-q_{1})} .
\end{eqnarray*}
Introducing the dependent variables as follows \cite{mr90j:58061}:
\begin{eqnarray*}
u=\frac{q_{t} \exp(q-q_1)}{1+\exp(q-q_1)},\qquad
v=\frac{q_{t}}{1+\exp(q-q_1)},
\end{eqnarray*}
then the equation can be written as 
\begin{eqnarray*}\label{rtoda}
\left\{ {\begin{matrix} u_t= u(u_{-1}-u_1+v-v_1)\\
   v_t= v(u_{-1}-u)  \end{matrix} } \right.
 \end{eqnarray*}
It admits two compatible Hamiltonian local structures 
\begin{equation*}
H_1=\left(\begin{matrix}
             0&u(1-\cS)\\(\cS^{-1}-1)u&u\cS-\cS^{-1}u
            \end{matrix}\right), \hspace{2 mm}
H_2=\left(\begin{matrix}
             u(\cS^{-1}-\cS)u&u(1-\cS)v\\v(\cS^{-1}-1)u&0
            \end{matrix}\right)
\end{equation*}
It is clear to see that ${\rm Ord} H_1=2$ and that the kernel of $H_1$ is 
spanned by $\left(\begin{matrix} \frac{1}{u} \\ 0 \end{matrix} \right)$ and 
$\left(\begin{matrix} 1 \\ 1 \end{matrix} \right)$. 
Similarly ${\rm Ord} H_2=2$ and the kernel of $H_2$ is spanned by 
$\left(\begin{matrix} \frac{1}{u} \\ 0 \end{matrix} \right)$ and 
$\left(\begin{matrix} 0 \\ \frac{1}{v} \end{matrix} \right)$. 
In other words, $H_1$ and $H_2$ have a common right divisor ${\rm Ord} H=1$
and can be written as $BH$ and $AH$ as follows:
\begin{eqnarray*}
&&H_1=BH=\left(\begin{matrix}
           0& u (1-\cS)\\1&u\cS-\cS^{-1}u
            \end{matrix}\right)\left(\begin{matrix}
            (1-S^{-1})u&0\\0&1
            \end{matrix}\right);\\
&&H_2=AH=\left(\begin{matrix}
            u(\cS+1)&u(1-\cS)v\\-v&0
            \end{matrix}\right)\left(\begin{matrix}
            (1-\cS^{-1})u&0\\0&1
            \end{matrix}\right) ,
\end{eqnarray*}
where $A$ and $B$ are of the total order $1$ and their kernels are of dimension $1$
Therefore both recursion operator
$R=AB^{-1}$ and its inverse $R^{-1}=BA^{-1}$ are weakly non-local, and
\begin{eqnarray*}
&&R=\left( {\begin{matrix}
u\cS+u+v_{1}+u_{1}+u\cS^{-1} & u\cS+u \\
 v+v\cS^{-1} & v
 \end{matrix} } \right)-\left(\begin{matrix} u_t \\v_t \end{matrix}\right) (\cS-1)^{-1}
 \left(\begin{matrix} \frac{1}{u} & 0 \end{matrix}\right);\\
&&R^{-1}= \left( \begin{matrix}
\frac{1}{v_1} &-\frac{u}{v_1^2} \cS+\frac{u}{v^2}-\frac{2u}{v v_1}\\
-\cS^{-1} \frac{1}{v}-\frac{1}{v_1}& \frac{u}{v_1^2} \cS+\cS^{-1} \frac{u}{v^2} +\frac{2 u}{v 
v_1}+\frac{1}{v}\end{matrix}
\right)+\left(\begin{matrix} \frac{u}{v_{1}} -\frac{u}{v} \\\frac{u_{-1}}{v_{-1}}-\frac{u}{v_1} \end{matrix}\right) 
(\cS-1)^{-1} \left(\begin{matrix} \frac{1}{u} & -\frac{2}{v} \end{matrix}\right).
\end{eqnarray*}

Note that the kernel of $A$ is spanned by $\left(\begin{matrix}0\\ 
\frac{1}{v}\end{matrix} \right)$, the kernel of $A^\dagger$ is spanned by 
$\left(\begin{matrix} \frac{1}{u} \\ -\frac{2}{v} \end{matrix} \right)$
and $B\left(\begin{matrix}0\\ \frac{1}{v}\end{matrix} 
\right)=\left(\begin{matrix} \frac{u}{v}-\frac{u}{v_{1}} \\ 
\frac{u}{v_1}-\frac{u_{-1}}{v_{-1}}\end{matrix} \right)$. This explains the 
nonlocal term 
in the inverse of the recursion operator.

\subsection{The Ablowitz-Ladik Lattice}\label{AL}
Consider the Ablowitz-Ladik Lattice \cite{ablowitz2}
\begin{eqnarray*}
\left\{ {\begin{matrix} u_t=(1-u v) (\alpha u_1-\beta  u_{-1}) \\ v_t= (1- uv) 
(\beta v_1-\alpha v_{-1})  \end{matrix} } \right. 
 \end{eqnarray*}
Its recursion operator \cite{mr93c:58096}
\begin{eqnarray}
&& R=\left(\begin{matrix} (1-uv)~\cS - u_1 v -u v_{-1} &-uu_{1}  \\ vv_{-1}&(1-uv)~\cS^{-1}\end{matrix} 
\right)+\left(\begin{matrix} -u\\v \end{matrix}\right) (\cS-1)^{-1}
\left(\begin{matrix} v_{-1} & u_{1} \end{matrix}\right) 
\nonumber\\&&\qquad -\left(\begin{matrix} (1-uv) u_{1}\\-(1-uv) v_{-1} \end{matrix}\right) (\cS-1)^{-1}
\left(\begin{matrix} \frac{v}{1-uv} & \frac{u}{1-uv} \end{matrix}\right)\label{real}
\end{eqnarray}
can be written as $R=AB^{-1},$
where by letting $w=1-u v$,  $w_i=\cS^i w$ and $p=u_1 v-u v_{-1}$ we have
\begin{eqnarray*}
&&A=\!\left(\!{\begin{matrix} w \cS \left(\frac{uv}{v_1}-u_1\right)
+\frac{u_1vw}{v_1}-\frac{u v_{-1}w_1}{v_1}
&w \cS \frac{u}{p} (1-\cS^{-1})-\frac{u^2 v_{-1}}{p}+\frac{u u_1 v}{p} \cS^{-1} 
\\
w \cS^{-1} (v_{-1}-\frac{v^2}{v_1})-u_1 v v_{-1}+\frac{u v^2 v_{-1}}{v_1}&
-w\cS^{-1}\frac{v}{p}(1-\cS^{-1})+\frac{uv v_{-1}}{p}-\frac{u_1 
v^2}{p}\cS^{-1}\end{matrix} } \!\right)
\end{eqnarray*}
and
\begin{eqnarray*}
 B=\left( {\begin{matrix}
 \frac{w}{v}\cS^{-1} p-u_1 w+\frac{uv w_1}{v_1}& \frac{u}{p}(1-\cS^{-1})\\v_{-1} 
w-\frac{v^2 w_1}{v_1}&-\frac{v}{p}(1-\cS^{-1})\end{matrix} } \right) .
\end{eqnarray*}

The operator $A$ can be factorized as follows:
\begin{eqnarray*}
\left( {\begin{matrix} 1 & 0\\(v_{-2}v-v_{-1}^2)w r \cS^{-1}& r\end{matrix}} 
\right)
       \left(\begin{matrix} \frac{q (u_1vw-u v_{-1}w_1)}{(uv-u_1v_1)}&  
-\frac{w}{(v_{-1}v_1w_1-v^2w)}
\cS-\frac{u_1vw-u v_{-1}w_1}{(uv-u_1v_1)(v_{-2}vw-v_{-1}^2w_{-1})}\\0& 1        
         \end{matrix}
 \right)D,
\end{eqnarray*}
where
\begin{eqnarray*}
&&r=\frac{1}{u v_{-1}w_{-1}-u_{-1}v_{-2}w}; \qquad q= 
\frac{1}{(v_{-1}v_1w_1-v^2w)(v_{-2}vw-v_{-1}^2w_{-1})};\\
&&D=\left(\begin{matrix} 0   & \hspace{-1.3cm} 
\begin{array}{c}v_{-1}w(v_{-1}v_1 
w_1-v^2w)\cS^{-2}-vw(v_{-2}v_1w_1-v_{-1}v w_{-1})\cS^{-1}\\+v_1w(v_{-2}v 
w-v_{-1}^2 w_{-1})\end{array}\\
 \frac{1}{v_1}(u_1v_1-uv)(v_{-2}v w-v_{-1}^2 w_{-1}) & \begin{array}{c}v_{-1}w 
\cS^{-2}+(u 
v_{-2}w-u_1v_{-1}w_{-1})\frac{v}{p}\cS^{-1}\\+(v_{-1}^2w_{-1}-v_{-2}v 
w)\frac{u}{p}\end{array}
         \end{matrix}
 \right).
\end{eqnarray*}
Note that ${\rm Ord} A={\rm Ord} D=2$ and $\ker D=\ker A$, which is spanned by 
$${\bf h}^{(1)}=\left(\begin{matrix} -\frac{1}{u_1 v-u v_{-1}}\\ \frac{v}{v_1}
\end{matrix}\right)\quad \mbox{and}\quad {\bf 
h}^{(2)}=\left(\begin{matrix}\frac{u v_1}{u_1v-u v_{-1}}\\ u_1v_1-1 
\end{matrix}\right).$$
Thus the operator $A$ is a full kernel operator and hence the inverse of $A B^{-1}$ 
is weakly nonlocal.
Note that $\ker D^{\dagger}$ is spanned by 
$${\bf g}^{(1)}=\left(\begin{matrix}vq\\0 
\end{matrix}\right)=\left(\begin{matrix}\frac{v}{(v_{-2}v w-v_{-1}^2 
w_1)(v_{-1}v_1w_1-v^2w)}\\0 \end{matrix}\right)
\quad \mbox{and}\quad {\bf g}^{(2)}=\left(\begin{matrix}-\frac{v_{-1}q}{w}\\0 
\end{matrix}\right)
$$ 
Thus $\ker A^{\dagger}$ 
is spanned by 
$$ \left(\begin{matrix} 1 & \cS w(v_{-1}^2-v_{-2}v)\\ 0 & \frac{1}{r}
       \end{matrix} \right)\left(\begin{matrix} \frac{(uv-u_1v_1)}{q (u_1vw-u 
v_{-1}w_1)}&0\\ \cS^{-1}\frac{w(uv-u_1v_1)(v_{-2}vw-v_{-1}^2w_{-1})}{u_1v w-u 
v_{-1}w_1}+v_{-1}v_1w_1-v^2w
&1\end{matrix}\right) {\bf g}^{(1)}
     =\left(\begin{matrix}
v_1\\u_{-1}\end{matrix}\right)$$
and similarly $\left(\begin{matrix} \frac{v}{1-uv} \\ 
\frac{u}{1-uv}\end{matrix}\right)$.
Moreover, we have
\begin{eqnarray*}
 B({\bf h}^{(1)})=\left(\begin{matrix}u\\ -v\end{matrix}\right); \qquad  B({\bf 
h}^{(2)})=\left(\begin{matrix} (1-uv)u_{-1}\\ -(1-uv) v_{1}
\end{matrix}\right).
\end{eqnarray*}
These give us the nonlocal term appearing in the inverse operator as stated in 
Theorem \ref{thmrw}, and indeed
\begin{eqnarray*}
&& R^{-1} =\left(\begin{matrix} (1-uv)~\cS^{-1}&uu_{-1}  \\ -vv_{1}& (1-uv)~\cS -uv_1-u_{-1}v \end{matrix} \right)
+\left(\begin{matrix} u\\-v \end{matrix}\right) (\cS-1)^{-1}
\left(\begin{matrix} v_1 & u_{-1} \end{matrix}\right)\\
&&\qquad +\left(\begin{matrix} (1-uv) u_{-1}\\-(1-uv) v_1 \end{matrix}\right) (\cS-1)^{-1}
\left(\begin{matrix} \frac{v}{1-uv} & \frac{u}{1-uv} \end{matrix}\right)
\end{eqnarray*} 

\subsection{The Kaup-Newell lattice}\label{seckn}
Consider the Kaup-Newell lattice \cite{tsuchida1}:
\begin{eqnarray*}
\left\{\begin{matrix}
 u_t=a\left( 
\frac{u_1}{1-u_{1}v_{1}}-\frac{u}{1-uv}\right)+b\left(\frac{u}{1+uv_{1}}-\frac{
u_{-1}}{1+u_{-1}v}\right)\\
 v_t=a\left( \frac{v}{1-u v}-\frac{v_{-1}}{1-u_{-1} 
v_{-1}}\right)+b\left(\frac{v_1}{1+uv_{1}}-\frac{v}{1+u_{-1}v}\right)
\end{matrix}  \right. :=a K_1 +b K_{-1}
\end{eqnarray*}
Its recursion operator 
\begin{eqnarray*}
&&R=\!\left(\!\!\!\! {\begin{matrix}
 -\frac{1}{(1-u_1v_1)^2} \cS+\frac{1}{(1-uv)^2}-\frac{2 u_1 v}{(1-u_1v_1)(1-uv)}
&-\frac{u_1^2}{(1-u_1 v_1)^2} \cS+\frac{u^2}{(1-u v)^2}-\frac{2 u u_1}{(1-uv)(1-u_1 v_1)} \\
-\frac{v_{-1}^2}{(1-u_{-1} v_{-1})^2} \cS^{-1}-\frac{v^2}{(1-u v)^2}&
-\frac{1}{(1-u_{-1}v_{-1})^2} \cS^{-1}+\frac{1-2uv}{(1-uv)^2}\end{matrix} } \!\!\!\!\right)
\\&&\qquad
 -2 K_1 (\cS-1)^{-1} 
\left(\begin{matrix} \frac{v}{1-u v} & \frac{u}{1-u v}\end{matrix}\right),
\end{eqnarray*}
can be written as $R=AB^{-1},$
where
\begin{eqnarray*}
&&A=\!\left(\!{\begin{matrix} 
(\cS-1)\frac{1}{v(1-uv)}(1-\cS^{-1})+2 (\cS-1)\frac{u}{1-uv}\cS^{-1}
&(\cS-1)\frac{u}{1-uv} \\
(1-\cS^{-1})\frac{v}{1-uv}(\cS^{-1}+1)&
(1-\cS^{-1})\frac{v}{1-uv}\end{matrix} } \!\right)
\\&& =
\left( {\begin{matrix} \cS-1 & 0\\0&1-\cS^{-1}\end{matrix}} \right)
\!\left(\begin{matrix}
 \frac{u}{1-uv}&0\\0&\frac{v}{1-uv}
\end{matrix}\right)\left(\begin{matrix}
2- \frac{1}{uv}&1\\1&1
\end{matrix}\right)\left(\begin{matrix}
 \cS^{-1}-1&0\\2&1
\end{matrix}\right)
\end{eqnarray*}
and
\begin{eqnarray*}
 B=\left( {\begin{matrix}
 \frac{1-u v}{v}(\cS^{-1}-1)&-u \\0&v\end{matrix} } \right) .
\end{eqnarray*}
The operator $A$ does not have a full kernel since ${\rm Ord} A=3$ and its 
kernel is spanned by $\left(\begin{matrix} 1\\ -2 \end{matrix} \right)$.
Surprisingly, operator $C=A-B$ can be factorised as follows:
\begin{eqnarray*}
 \left(\begin{matrix} 1& \frac{1}{v_1^2}\cS\\0 &1 
       \end{matrix} \right)\left(\begin{matrix} 1&0\\ (uv-\cS^{-1}) 
\frac{v_1^2}{1-u^2 v_1^2} &1\end{matrix}\right)
       \left(\begin{matrix} \frac{1+u v_1}{v_1^2(1-uv)}& 0\\0& 1        
         \end{matrix}
 \right)D,
\end{eqnarray*}
where
\begin{eqnarray*}
 D=\left(\begin{matrix} (v-2v_1+uv v_1)+v(1-u v_1) \cS^{-1}& v(1-u^2 v_1^2)
 \\v(1-\cS^{-1})\frac{1+u v_1}{1-u v_1} & 0        
         \end{matrix}
 \right).
\end{eqnarray*}
Note that ${\rm Ord} C={\rm Ord} D=1$ and $\ker D=\ker C$, which is spanned by 
${\bf h}=\left(\begin{matrix}\frac{1-u 
v_1}{1+uv_1}\\\frac{2v_1}{v(1+u v_1)}-\frac{2}{(1+u_{-1}v)} 
\end{matrix}\right)$.
Thus operator $C$ is a full kernel operator and hence the inverse of $(A-B) 
B^{-1}$ is weakly nonlocal as presented in \cite{kmw13} and it equals to
\begin{eqnarray}
&&\!\left(\!\!\!\! {\begin{matrix} \frac{1}{(1+u_{-1}v)^2} \cS^{-1} -\frac{1+2 u v_1}{(1+u v_1)^2} & 
-\frac{u^2}{(1+u v_1)^2} \cS
+\frac{u_{-1}^2}{(1+u_{-1} v)^2}-\frac{2 u u_{-1}}{(1+u_{-1}v) (1+u v_1)}\\
-\frac{v^2}{(1+u_{-1} v)^2} \cS^{-1}-\frac{v_1^2}{(1+u v_1)^2} &
\frac{1}{(1+u v_1)^2} \cS-\frac{1}{(1+u_{-1} v)^2}-\frac{2 u_{-1} v_1}{(1+u_{-1} v) (1+u v_1)}
\end{matrix} } \!\!\!\!\right)
\nonumber\\&&\qquad
 -2 K_{-1} (\cS-1)^{-1} 
\left(\begin{matrix} \frac{v_1}{1+u v_1} & \frac{u_{-1}}{1+u_{-1} v}\end{matrix}\right). \label{inkn}
\end{eqnarray}
Note that $\ker D^{\dagger}$ is spanned by $\left(\begin{matrix}0\\\frac{1}{v} 
\end{matrix}\right)$ and thus $\ker C^{\dagger}$ 
is spanned by 
$$ \left(\begin{matrix} 1& 0\\-\cS^{-1} \frac{1}{v_1^2} &1 
       \end{matrix} \right)\left(\begin{matrix} 1&-\frac{v_1^2}{1-u^2 v_1^2} 
(uv-\cS)\\ 0 &1\end{matrix}\right)
       \left(\begin{matrix} \frac{v_1^2(1-uv)}{1+u v_1}& 0\\0& 1        
         \end{matrix}
 \right)\left(\begin{matrix}0\\\frac{1}{v} 
\end{matrix}\right)=\left(\begin{matrix}\frac{v_1}{1+u 
v_1}\\\frac{u_{-1}}{1+u_{-1}v} \end{matrix}\right).$$
Moreover, we have
\begin{eqnarray*}
 B({\bf h})=2 \left(\begin{matrix}\frac{u}{1+u 
v_1}-\frac{u_{-1}}{1+u_{-1}v}\\\frac{v_1}{1+u v_1}-\frac{v}{1+u_{-1}v} 
\end{matrix}\right)=2 K_{-1}.
\end{eqnarray*}
These give us the nonlocal term appearing in the inverse operator as shown in 
(\ref{inkn}).

\section{Conclusions}
In this paper we have built  a rigorous algebraic setting for difference and
rational (pseudo--difference) operators with coefficients in a difference field 
$\cF$
  and  study their properties. In particular, we formulate a criteria for a 
rational operator to be weakly nonlocal. We have defined and studied 
preHamiltonian pairs, which is a generalization of the well known bi-Hamiltonian structures in the theory 
of integrable systems.
 By definition a preHamiltonian operator is an
operator whose images form a Lie subalgebra in the Lie algebra of evolutionary 
derivations of $\cF$. The latter can be directly verified and it is a 
relatively simple problem comparing to the verification of the Jacobi identity 
for Hamiltonian operators. We have shown that a recursion Nijenhuis operator 
is a ratio of difference operators from  a preHamiltonian pair.
Thus for a given rational operator, to test 
whether it is Nijenhuis or not can be done
systematically. We applied our theoretical results to integrable 
differential difference equations in two aspects: 
\begin{itemize}
 \item We have constructed
a rational recursion operator $R$ (\ref{readler}) for Adler--Postnikov 
integrable equation (\ref{bv}) and shown that it can be written as the ratio of a preHamiltonian pair and thus it
is Nijenhuis. Moreover, we proved that $R$ produces infinitely many 
commuting local 
symmetries;
\item For a given recursion operator we can answer the question whether the 
inverse  operator is
weakly nonlocal  and, if so, how to bring it to the standard weakly nonlocal 
from (examples in Section 6).
\end{itemize}

In section \ref{seckn} we show that for a weakly nonlocal recursion operator $R$ 
which does not have a weakly nonlocal inverse, may exist a constant 
$\gamma\in\cK$ such that $(R-\gamma {\,\rm  Id})^{-1}$ is weakly nonlocal. In 
other words, the total order of the difference operator $A-\gamma B$ in the 
factorisation $R=AB^{-1}$ may be lower for a certain choice of 
$\gamma$. This observation requires further investigation.

The concept of preHamiltonian operators deserves further attention. These 
operators naturally appear in the description of the invariant evolutions 
of curvature flows 
in homogeneous spaces in both continuous\cite{mr2003f:37137} and discrete 
\cite{Manbeffawang13} setting. In the future, we'll look into the geometric 
implication of such operators. 

In this paper, we mainly explored the relation between PreHamiltonian operators and Nijenhuis operators.
We are going to investigate how preHamiltonian pairs relate to 
biHamiltonian pairs. In our forthcoming paper \cite{CMW18}, we'll present the following main result:
{\it if 
$H$ is a Hamiltonian 
(a priori non-local, i.e. rational) operator, then to find a second Hamiltonian 
$K$ compatible with $H$ is the same as to find a preHamiltonian pair $A$ and $B$ 
such that $AB^{-1}H$ is anti-symmetric}.

We have discovered that Adler--Postnikov 
integrable equation \eqref{bv} is indeed a Hamiltonian system. This equation can be written as 
$u_t=H \delta_u (\ln u)$, where $H$ is the following anti-symmetric rational operator
\begin{equation*} \label{ham}
\begin{split}
H&=u^2u_1u_2^2\cS^2-\cS^{-2}u^2u_1u_2^2+\cS^{-1}uu_1(u+u_1)-uu_1(u+u_1)\cS \\
&+u(1-\cS^{-1})(1-uu_1)({\cS}u-u\cS^{-1})^{-1}(1-uu_1)(\cS-1)u.
\end{split}
\end{equation*}
In \cite{CMW18}, we are going to show that $H$ is a 
Hamiltonian operator for equation \eqref{bv} and explain how it is related to 
the recursion operator (\ref{readler}).

\section*{Appendix A. Basic concepts for a unital associative ring}

Recall definitions of some basic concepts for a unital associative ring $\cR$ 
(see for example \cite{CDSK2013}). 

A {\sl left} (respectively {\sl right}) {\sl 
ideal} of  $\cR$ is an additive subgroup $\cI\subset \cR$ such that $\cR 
\cI=\cI$ (resp. $\cI\cR=\cI$).

A left (resp. right) {\sl principal} ideals generated by $a\in\cR$ is, 
by definition, $\cR a$ (resp. $a\cR$).

A ring is 
called a {\sl principal ideal ring}, if every left and right ideal of the ring 
is 
principal.

Given an element $a\in\cR$, an element $d$ is called a {\em right} (resp. 
{\em left}) {\em divisor} of $a$ if $a=b d$ (resp. $a=db $) for some 
$b\in\cR$. An element $m\in\cR$ is called {\em left} (resp. {\em right}) 
{\em multiple} of $a$ if $m=ba$ (resp. $m=ab$) for some $b\in\cR$.

Given elements $a, b \in \cR$, their {\em right} (resp. {\em left}) 
{\em greatest common divisor} ({\em gcd}) is the generator $d$ of the {\em 
left} (resp. {\em right})
ideal generated by $a$ and $b$: $\cR a + \cR b = \cR d$ (resp. $a \cR + b\cR = 
d\cR$). It is uniquely defined up to multiplication by an
invertible element. It follows that $d$ is a right (resp. left) divisor of 
both $a$ and $b$, and we have the {\em Bezout identity}
$d = ua + vb$ (resp. $d = au + bv$) for some $u, v \in \cR$.

Similarly, the {\em left} (resp. {\em right}) {\em least common multiple} ({\em 
lcm}) of $a$ and $b$ is an element $m \in \cR$, defined, uniquely up to
multiplication by an invertible element, as the generator of the intersection 
of 
the {\em left} (resp. {\em right}) principal ideals
generated by $a$ and by $b$: $\cR m = \cR a \cap \cR b$ (resp. $m\cR = a\cR 
\cap b\cR$).

We say that $a$ and $b$ are {\em right} (resp. {\em left}) {\em coprime} if 
their {\em right} (resp. {\em left}) greatest common divisor is 1 (or
invertible), namely if the left (resp. right ) ideal that they generate is the 
whole ring $ \cR a + \cR b =\cR$ (resp.
$a\cR + b\cR = \cR$). 

An element $a\in\cR\setminus \{0\}$ is called a {\em right zero divisor} if 
there exist $b\in\cR\setminus \{0\}$ (called {\em a left zero divisor}) such 
that $ba=0$. 

A non-zero element $a\in\cR$ is called {\em regular} if it is neither a left 
nor a right zero divisor. A set of regular elements $\cR^\times=\{a\in\cR\,|\, 
a \ \mbox{is regular}\}$ is a multiplicative monoid of $\cR$.

Ring $\cR$ is called a {\sl domain}, if it does not have zero divisors.

A domain $\cR$ is called {\sl right (left) Euclidean}, if there exist a 
function 
\[
 {\rm Ord}:\, \cR\setminus \{0\}\mapsto\bbbz_{\ge 0},
\]
 such that 
\begin{enumerate}
 \item $ {\rm Ord}\, (a)\le {\rm Ord}\, (ab)\ge {\rm Ord}\, (b),\quad \forall 
a,b\in \cR\setminus \{0\}$,
 \item for any $a,b\in\cR,\ b\ne 0$ there exist unique $c_r,q_r\in\cR$
(resp.  $c_l,q_l\in\cR$), such that 
 \[
a=b c_r+q_r=c_l b+q_l
 \]
and $q_r=0$ or ${\rm 
Ord}\, q_r< {\rm Ord}\, b$ (resp.  $q_l=0$ or ${\rm Ord}\, q_l< {\rm Ord}\, 
b$).
\end{enumerate}

Ring $\cR$ satisfies the right (left) Ore property if for any $a\in \cR,\ 
b\in\cR^\times$ there exist $c\in \cR^\times, d\in \cR$ (resp.  $c_1\in 
\cR^\times, d_1\in \cR$) such that $ac=bd$ (resp. $c_1 a=d_1 b$).

\section*{Appendix B. Lemmas used for the proof of Proposition \ref{propr}}
We denote by $\pi$ the projection from the space of Laurent difference polynomials $\mathrm{A}$ to the space of difference polynomials $\mathrm{K}$ defined by letting $\pi(b)$ being the nonsingular part of $b$ for all difference Laurent monomial $b \in \mathrm{A}$. For example,
\begin{equation*}
\pi(u+\frac{uu_1}{u_2})=u.
\end{equation*}
If $L=\sum_{n \leq N}{l_n S^n}$ is a Laurent series with coefficients being 
Laurent difference polynomials, we denote by $\pi(L)$ the series $\sum_{n \leq 
N}{\pi(l_n) S^n}$.

\begin{Lem} \label{lem5}
Let $a,b,c,d \in \mathrm{K}$ and $n \in \bbbz$. Then $\pi[(a+bS^{-1})\Delta^{-1}(c+dS^{-1})]$ is a difference operator. 
\end{Lem}
\begin{proof}
We have 
\begin{equation}
\begin{split}
(a+bS^{-1})\Delta^{-1}(c+dS^{-1})&= \sum_{n \geq 1}{(ac_{-2n-1} \beta^n+bd_{-2n}(\beta^{n-1})_{-1})S^{-2n-1}}\\
&+\sum_{n \geq 0}{(ad_{-2n-1} \beta^n+bc_{-2n-2}(\beta^n)_{-1})S^{-2n-2}}\\
&+ac_{-1}\beta^0 S^{-1}
\end{split}
\end{equation}
It is clear that for $n$ large enough $\pi(ac_{-2n-1} \beta^n)=0$ and similarly $\pi(bd_{-2n}(\beta^{n-1})_{-1})=0$. 
\begin{equation*}
\pi(ac_{-2n-1} \beta^n+bd_{-2n}(\beta^{n-1})_{-1})=0 \text{  for } n \gg 0.
\end{equation*}
Similarly,
\begin{equation*}
\pi(ad_{-2n-1} \beta^n+bc_{-2n-2}(\beta^n)_{-1})=0 \text{  for } n \gg 0.
\end{equation*}
\end{proof}

\begin{Lem}\label{lem6}
Let $a,b,c,d, \in \mathrm{K}$ and $e \in \mathrm{A}$. Then $\pi[(a+bS^{-1})\Delta^{-1}eS^{-1} \Delta^{-1}(c+dS^{-1})]$ 
is a difference operator.
\end{Lem}
\begin{proof}
Let us expand $L=(a+bS^{-1})\Delta^{-1}eS^{-1} \Delta^{-1}(c+dS^{-1})$ as a Laurent series in $S^{-1}$:
\begin{equation}\label{lem6a}
\begin{split}
L=&(a+bS^{-1})(\sum_{n \geq 0}{\beta^ne_{-2n-1} S^{-2n-1}})(\sum_{k \geq 
0}{(\beta^k)_{-1}S^{-2k-1}})(c_{-1}S^{-1}+d_{-1}S^{-2}) \\
=&(a+bS^{-1})(\sum_{m \geq 0}{\beta^{m+1}(\sum_{n= 0}^m{(\frac{e}{u})_{-2n-1}})S^{-2m-2}})(c_{-1}S^{-1}+d_{-1}S^{-2})\\
=& \sum_{m \geq 0}{ac_{-2m-3}\beta^{m+1}(\sum_{n= 0}^m{(\frac{e}{u})_{-2n-1}})S^{-2m-3}} \\
+& \sum_{m \geq 0}{bd_{-2m-4}(\beta^{m+1})_{-1}(\sum_{n= 0}^m{(\frac{e}{u})_{-2n-2}})S^{-2m-5}} \\
+& \sum_{m \geq 0}{bc_{-2m-4}(\beta^{m+1})_{-1}(\sum_{n= 0}^m{(\frac{e}{u})_{-2n-2}})S^{-2m-4}} \\
+& \sum_{m \geq 0}{ad_{-2m-3}\beta^{m+1}(\sum_{n= 0}^m{(\frac{e}{u})_{-2n-1}})S^{-2m-4}}
\end{split}
\end{equation}
After applying $\pi$ to the coefficients of this Laurent series expansion of $L$, we get a difference operator. Let us 
show it for the first summand in the last line of \eqref{lem6a}, namely that 
\begin{equation*}
 \sum_{m \geq 0}{\pi(ac_{-2m-3}\beta^{m+1}(\sum_{n= 0}^m{(\frac{e}{u})_{-2n-1}}))S^{-2m-3}} 
\end{equation*}
is a difference operator (the same argument applies to the remaining three summands). This follows from the claim that 
for large enough $m$, and for all $0 \leq n \leq m$,
\begin{equation}\label{lem6z}
\pi(ac_{-2m-3}\beta^{m+1}(\frac{e}{u})_{-2n-1})=0.
\end{equation}
Indeed, if $e$ can be written as a sum of Laurent monomials for which the degree of the numerators, as polynomials in 
the $u_i$'s are bounded by $m_e$, and if $m_a$ and $m_c$ denote the degrees of $a$ and $c$ as polynomials in the 
$u_i$'s, then \eqref{lem6z} holds for $m>m_a+m_c+m_e$.
\end{proof}

\begin{Lem}\label{lem3}
Let $f$ be a difference polynomial such that $R$ is recursion for the equation 
$u_t=f$. Then there exists
a difference polynomial $k$ such that $f=u(k_2-k)$.
\end{Lem}
\begin{proof}
Operator $R$ given by (\ref{readler}) is recursion for $u_t=f$ which implies 
that $\ln(u)$ is a conserved density of 
$f$, or in other words that there is a difference polynomial
$g$ such that $f=u(g_1-g)$.
\\
To conclude we need to prove that $g_1-g=k_2-k$ for some difference polynomial 
$k$, which is equivalent to say that 
$g=k_1+k+\rho$ for some constant $\rho$. We claim that this is the same as 
saying that
\begin{equation} \label{twistedvd}
\sum_n{(-1)^n S^{-n}(\frac{\partial g}{\partial u_n})}=0.
\end{equation}
Indeed, it is clear that $\sum_n{(-1)^n S^{-n}\frac{\partial }{\partial 
u_n}(S+1)}=0$ by \eqref{spart} and that any 
constant satisfies \eqref{twistedvd}. Conversely, if a difference polynomial 
$g$ 
of order $(M,N)$ satisfies 
\eqref{twistedvd}, then there exists 
a difference polynomial $k$ and a constant $\rho$ such that $g=k_1+k+\rho$. To 
check this, we proceed by induction on 
the total order of $g$. If it is zero, meaning  that $g$ is a function of $u_N$ 
for a single $N$, then $g$ must be a 
constant.
If not, say if $g$ has order $(M,N)$ with $M<N$, then $\frac{\partial 
g}{\partial u_N}$ does not depend on $u_M$. Consequently, we can write $g$ as a 
sum $h+k$ where $k$ has order $(M',N)$ with $M<M'$ and $h$ has order $(M, N')$ 
with $N'<N$. Since $g$ and $k+k_{-1}$ both satisfy \eqref{twistedvd}, it 
follows 
that $h-k_{-1}$ must satisfy \eqref{twistedvd} as well, i.e. we
reduced the problem to a difference polynomial of lesser total order. 
\\
The difference polynomial \eqref{twistedvd} is the remainder of the division of 
$g_*$ by 
$(S+1)$ on the left. Let us call it $r$:
\begin{equation} \label{defr}
g_*=(S+1)X+r, \hspace{2 mm} r=\sum_n{(-1)^n S^{-n}(\frac{\partial g}{\partial 
u_n})},
\end{equation}
where $X$ is some difference operator.
We want to prove that $r=0$. It is equivalent to prove that the remainder $r'$ 
of 
the 
division of 
$g_* u (S^2-1)$ by $S+1$ on the left is $0$. Indeed $r'=ur-(ur)_{-2}$ and $r$ 
is a difference polynomial, therefore 
$r=0 \iff r'=0$. 
\par
We are going to deduce that $r'=0$ from the fact that
$R$ is recursion for $f=u(g_1-g)$. Note that $f_*=u(S-1)g_*+g_1-g$. Recall 
equation \eqref{readler1} where $R$ was 
expressed as $(Q\Delta^{-1}C+P)(S^2-1)^{-1}\frac{1}{u}$. By definition 
\eqref{reopev} of a recursion operator we have
\begin{equation} \label{lem2a}
\begin{split}
&(Q_*[f]-f_* Q)\Delta^{-1}C-Q\Delta^{-1}\Delta_*[f] 
\Delta^{-1}C\\
&+Q\Delta^{-1}C_*[f] +P(S+1)^{-1}g_* u (S^2-1)\\
&-f_* P+ P_*[f]+Q \Delta^{-1}C (S+1)^{-1} g_* u (S^2-1)=0.
\end{split}
\end{equation}
The idea is to expand \eqref{lem2a} as a Laurent series in $S^{-1}$ and to 
project the coefficients in front of $S^{-N}$ for large $N$ on the space of 
difference polynomials. Let us start by rearranging \eqref{lem2a} using two 
Euclidean divisions
\begin{equation}\label{ed}
C= w_2+w_{-1}+Z(S+1),\qquad
g_* u (S^2-1)=r'+(S+1)Y,
\end{equation}
where $Y$ and $Z$ are two difference operators. Combining \eqref{lem2a} with 
\eqref{ed}, we get:
\begin{equation}\label{lem2d}
\begin{split}
&Q \Delta^{-1}(w_2+w_{-1})S^2(S+1)^{-1}r'+p(S+1)^{-1}r'\\
&=Q\Delta^{-1}\Delta_*[f] 
\Delta^{-1}C-(Q_*[f]-f_* Q)\Delta^{-1}C\\
&-Q\Delta^{-1}(C_*[f]+CY+Zr')-Pr'
+f_* P- P_*[f]
\end{split}
\end{equation}
\iffalse
Let $\mathrm{A}=\cK[...,u_{-1},u_{-1}^{-1},u,u^{-1},u_1^{-1},u_1,...]$ be the 
space of Laurent polynomials in the $u_i$'s and let $\pi$ be the projection 
$\pi: \mathrm{A} \mapsto \mathrm{K}$ defined by letting $\pi(b)$ being the 
nonsingular part of the Laurent difference polynomial $b$ for all $b \in 
\mathrm{A}$. For instance,
\begin{equation*}
\pi(u_2+\frac{uu_{-1}}{u_1})=u_2.
\end{equation*}
If $L=\sum_{n \leq N}{l_n S^n}$ is a Laurent series with coefficients being 
Laurent difference polynomials, we denote by $\pi(L)$ the series $\sum_{n \leq 
N}{\pi(l_n) S^n}$.
\fi
By Lemma \ref{lem5} and Lemma \ref{lem6}, if $M$ is the RHS 
of \eqref{lem2d}, $\pi(M)$ is a difference operator. Therefore,
\begin{equation} \label{lem2b}
\pi[ Q \Delta^{-1}(w_2+w_{-1})S^2(S+1)^{-1}r'+p(S+1)^{-1}r'] 
\end{equation}
must be a difference operator as well.
Let us write $Q=a+bS^{-1}$ where $a=u(uu_1-1)$ and $b=u(1-uu_{-1})$ and let 
$c=w_2+w_{-1}$. Looking only at even powers of 
$S^{-1}$ in the Laurent series expansion of (\ref{lem2b}) we obtain
\begin{equation}
\pi[(a(\beta^0c_{-1}+...+\beta^Nc_{-2N-1})-b(\beta^0c_{-1}+...+\beta^{N-1}c_{
-2N+1})_{-1}-p){r'}_{-2N}]=0 \text{ for all 
} N \gg 0,
\end{equation}
where the Laurent difference polynomials 
$\beta^n=\frac{u_{-1}...u_{-2n-1}}{u...u_{-2n}}, n \geq 1$, 
$\beta^0=\frac{1}{u}$ satisfy
\begin{equation}\label{invofdelta}
\Delta^{-1}=\sum_{n \geq 0}{\beta^nS^{-2n-1}}.
\end{equation}
It is clear that for all $k> 1$ and for all $N \geq \Ord  \, r' +2$, we have
$$\pi(ac_{2k-1}\beta^k {r'}_{-2N})=\pi(b(c_{2k-1}\beta^k)_{-1}{r'}_{-2N})=0.$$ 
In other words, there 
exists 
$K \geq 0$ such that
\begin{equation}
\pi[(a(\beta^0c_{-1}+...+\beta^Kc_{-2K-1})-b(\beta^0c_{-1}+...+\beta^{K-1}c_{
-2K+1})_{-1}-p){r'}_{-2N}]=0 \text{ for all 
} N \gg 0.
\end{equation}
If $r' \neq 0$, $r'$ is either a constant or the order of ${r'}_{-2N}$ must go 
to $(-\infty,-\infty)$ as $N$ grows. In 
both cases we must have:
\begin{equation}
\pi[a(\beta^0c_{-1}+...+\beta^Kc_{-2K-1})-b(\beta^0c_{-1}+...+\beta^{K-1}c_{
-2K+1})_{-1}-p]=0.
\end{equation}
This quantity can be computed directly, and we obtain
\begin{equation}
\begin{split}
p &=-2+u(u_2+3u_1+2u+u_{-1}+u_{-2})+2u_{-1}u_{-3} \\
&-u(2u_1u_{-1}u_{-3}+2u_1uu_{-1}+uu_1u_2+u_{-2}u_{-1}u+2u_{-4}u_{-2}u)
\end{split},
\end{equation}
which is a contradiction to $p$ given in (\ref{pp}). Thus we have $r'=0$ and 
hence $g=k_1+k+\rho$. By now we have
proved the statement.
\end{proof}

\begin{Lem}\label{lem4}
Let $g \in \mathrm{K}$ be such that $R$ is recursion for 
$f=u(g_2-g)$. Then
\begin{equation*}
Q \Delta^{-1} \left(C(g)_*u(S^2-1)-C(g)(S^2-S)\right)
+\left(Q_*[f]-f_*Q-Q(g_1-g_2)\right)\Delta^{-1}C
\end{equation*}
is a difference operator.
\end{Lem}
\begin{proof}
We have $\Delta^{-1}_*[f]=(g_1-g_2)\Delta^{-1}+\Delta^{-1}(g_1-g_2)$ and
$f_*=u(S^2-1)g_*+g_2-g$. From \eqref{lem2a}  we deduce that 
\begin{equation} \label{eq4}
Q \Delta^{-1}\left(C g_*u(S^2-1)-(g_2-g_1)C+C_*[f]\right)
+\left(Q_*[f]-f_*Q-Q(g_2-g_1)\right)\Delta^{-1}C
\end{equation}
is a difference operator. It remains to rewrite the first non-local term. We 
have modulo left multiplication by $\Delta$ and we have
\begin{equation*}
\begin{split}
C g_*u(S^2-1)&=C(g)_*u(S^2-1)-(g_2(w_2)_*-g_1(w_{-1})_*)u(S^2-1) \\             
            &=C(g)_*u(S^2-1)+u_1u_3g_2(S^5-S)-uu_{-2}g_1(S^2-S^{-2}) \\
                          & \equiv 
C(g)_*u(S^2-1)+(uu_{-2}g_{-2}-u_1u_3g_2)S-u_1u_3(g_5-g_1)S^2
\end{split}
\end{equation*}
and
\begin{equation*}
C_*[f]=uu_{-2}(g_2-g_{-2})S-u_1u_3(g_5-g_1)S^2 .
\end{equation*}
Therefore
\begin{equation}\label{eq5}
C g_*u(S^2-1)-(g_2-g_1)C+C_*[f] \equiv C(g)_*u(S^2-1)-C(g)(S^2-S).
\end{equation}
We conclude combining \eqref{eq4} to \eqref{eq5}.
\end{proof}

\begin{Lem}\label{lem7}
Let $a,b, c,d,e,f,g,h$ be difference Laurent polynomials such that $a, b, g, h \neq 0$ and
\begin{equation*}
(a+bS^{-1})\Delta^{-1}(c+dS^{-1})+(e+fS^{-1})\Delta^{-1}(g+hS^{-1})
\end{equation*}
is a difference operator. Then there exists a constant $\lambda \in \cK$ such that
\begin{equation*}
\begin{split}
e+fS^{-1}&=\lambda (a+bS^{-1})\\
c+dS^{-1}&=-\lambda(g+hS^{-1})
\end{split}
\end{equation*}
\end{Lem}
\begin{proof}
Recall the definition of the Laurent monomials $\beta^n$ for $n \geq 0$ 
\begin{equation}
\Delta^{-1}=\sum_{n \geq 0}{\beta^n S^{-2n-1}}.
\end{equation}
We have 
\begin{equation}
\begin{split}
(a+bS^{-1})\Delta^{-1}(c+dS^{-1})&= \sum_{n \geq 1}{(ac_{-2n-1} \beta^n+bd_{-2n}(\beta^{n-1})_{-1})S^{-2n-1}}\\
&+\sum_{n \geq 0}{(ad_{-2n-1} \beta^n+bc_{-2n-2}(\beta^n)_{-1})S^{-2n-2}}\\
&+ac_{-1}\beta^0 S^{-1}
\end{split}
\end{equation}
Therefore, we must have
for large enough $n$
\begin{equation*}
\begin{split}
\beta^{n-1}(ad_{-2n+1}+eh_{-2n+1})+(\beta^{n-1})_{-1}(bc_{-2n}+fg_{-2n})&=0 \\
\beta^{n}(ac_{-2n-1}+eg_{-2n-1})+(\beta^{n-1})_{-1}(bd_{-2n}+fh_{-2n})&=0.
\end{split}
\end{equation*}
Here $\beta^n$ has poles at $u,u_{-2},...,u_{-2n}$ 
($\beta^n=\frac{u_{-1}...u_{-2n+1}}{u...u_{-2n}}$) and 
$(\beta^{n-1})_{-1}$ has poles 
at $u_{-1},...,u_{-2n+1}$. Moreover, the Laurent polynomials inside the 
parenthesis can only have a bounded number of poles, independently of $n$. Combining these two facts we 
deduce 
that for large $n$ the arguments inside 
the four parenthesis 
must vanish:
\begin{equation}
\begin{split}
0&=ad_{-2n+1}+eh_{-2n+1} \\
0&=bc_{-2n}+fg_{-2n}\\
0&=ac_{-2n-1}+eg_{-2n-1}\\
0&=bd_{-2n}+fh_{-2n}, \hspace{2 cm} n \gg 0.
\end{split}
\end{equation}
Since $a,b,g,h \neq 0$, either $e=f=c=d=0$, in which case we can take $\lambda=0$, or $e,f,c,d, \neq 0$. In the latter case we conclude using the fact that, if two Laurent difference polynomials $x$ and $y$ are such that $x_{2n}=y$ for infinitely many $n \in \bbbz$, then $x$ and $y$ are both equal to the same constant.
\end{proof}

\begin{Lem}\label{lem8}
Let $d$ be a difference polynomial. Then $d$ is in the image of $\Delta$ if and 
only if
\begin{equation}
   \d_* u (1+S^{-1})-d=\Delta P,
\end{equation}
where $P$ is a difference operator. In this case, we have
\begin{equation}
d=\Delta(-S(\sum_{n}{\frac{\alpha_{2n}}{u}(\frac{\partial d}{\partial 
u_{2n}})_{-2n}})).
\end{equation}
Here for all $n \in \bbbz$,  $\alpha_{2n}$ (resp. $\alpha_{2n+1}$) is the 
unique 
difference Laurent polynomial such that
$S^{2n}u-\alpha_{2n}$ (resp. $S^{2n+1}u-\alpha_{2n+1}S^{-1}$) is divisible on 
the left by $\Delta$. Moreover,
\begin{equation*}
\sum_{n}{\frac{\alpha_{2n}}{u}(\frac{\partial d}{\partial u_{2n}})_{-2n}}
\end{equation*}
is a difference polynomial.
\end{Lem}
\begin{proof}
Suppose that $d=u_1d'_1-ud'_{-1}$ for a difference Laurent polynomial $d'$. 
Then 
the Fr{\'e}chet derivative of $d$  
expands as:
\begin{equation*}
\d_*=\Delta \d'_*+ d'_1S-d'_{-1}.
\end{equation*}
Hence (we use to $\equiv$ to denote modulo left multiplication by $\Delta$) we 
get
\begin{equation*}
\d_*u(1+S^{-1}) \equiv (Sd'-d'_{-1})u(1+S^{-1}) \equiv u_1d'_1-ud'_{-1} \equiv d
\end{equation*}
Conversely assume that 
\begin{equation} \label{eq1}
\d_* u (1+S^{-1}) \equiv d.
\end{equation}
Recall that the $\alpha_n$'s are defined so that $S^{2n}u \equiv \alpha_{2n}$ 
and $S^{2n+1}u \equiv \alpha_{2n+1} S^{-1}$ for all $n \in \bbbz$. The 
following 
identity can be easily checked by induction
\begin{equation} \label{alphas}
\alpha_{2n+2}=\frac{u}{u_{-1}} 
\alpha_{2n+1}=\frac{u^2}{u_{-1}u_{-2}}(\alpha_{2n})_{-2}, \hspace{2 mm} \forall 
n \in \bbbz.
\end{equation} 
Let us rewrite the LHS of \eqref{eq1}:
\begin{equation} \label{eq2}
\begin{split}
 \sum_{n}{S^n(\frac{\partial d}{\partial u_n})_{-n}u(1+S^{-1})}&=\sum_{n}{S^n 
u((\frac{\partial d}{\partial u_n})_{-n}+\frac{u_1}{u}(\frac{\partial 
d}{\partial u_{n+1}})_{-n})}\\
&\equiv \sum_{n}{\alpha_{2n}((\frac{\partial d}{\partial u_{2n}})_{-2n}+ 
\frac{u_1}{u} (\frac{\partial d}{\partial u_{2n+1}})_{-2n})}\\
&+\sum_{n}{\alpha_{2n+1}S^{-1}((\frac{\partial d}{\partial u_{2n+1}})_{-2n-1}+ 
\frac{u_1}{u} (\frac{\partial d}{\partial u_{2n+2}})_{-2n-1})}
\end{split}
\end{equation}
Combining \eqref{eq1}, \eqref{alphas} and \eqref{eq2} we obtain
\begin{equation} \label{eq3}
\begin{split}
d&= \sum_{n}{\alpha_{2n}(\frac{\partial d}{\partial 
u_{2n}})_{-2n}}+\frac{u_1}{u}\sum_{n}{\alpha_{2n}(\frac{\partial 
d}{\partial u_{2n+1}})_{-2n}}, \\
0&=\sum_{n}{\frac{u}{u_{-2}}(\alpha_{2n})_{-2}(\frac{\partial d}{\partial 
u_{2n+1}})_{-2n-2}}+\sum_{n}{\alpha_{2n+2}(\frac{\partial d}{\partial 
u_{2n+2}})_{-2n-2}}.
\end{split}
\end{equation} 
from which it follows that 
\begin{equation*}
\begin{split}
d&=\Delta(-\sum_{n}{\frac{(\alpha_{2n})_{-1}}{u_{-1}}(\frac{\partial 
d}{\partial 
u_{2n+1}})_{-2n-1}})\\
& =\Delta(-S(\sum_{n}{\frac{\alpha_{2n}}{u}(\frac{\partial d}{\partial 
u_{2n}})_{-2n}})).
\end{split}
\end{equation*}
We proved that there exists a Laurent difference polynomial $d'$ such that 
$d=u_1d'_1-ud'_{-1}$. It implies that
$d'$ cannot have poles (since its highest pole should be lesser or equal than 
$0$ and its lowest pole should be greater 
than $0$), therefore that it is a difference polynomial.
\end{proof}

\begin{Lem} \label{lem9}
Let $A$ and $B$ be two nonzero left coprime difference operators with coefficients in $\cF$. Suppose that $A(x)=B(y)$ for some $x,y \in \cF$. Let $M=AC=BD$ be there right least common multiple. Then, there exists $z \in \cF$ such that
$x=C(z)$ and $y=D(z)$. In particular $Im A \cap Im B =Im M$.
\end{Lem}
\begin{proof}
By definition of $M$, $C$ and $D$ are right coprime. Hence we can consider a Bezout identity
\begin{equation} \label{bezcd}
UC+VD=1
\end{equation}
for two difference operators $U$ and $V$. After replacing $U$ by $U+\lambda D$ and $V$ by $V-\lambda C$ for $\lambda \in \cK$, \eqref{bezcd} still holds. Hence we can assume that both $U$ and $V$ are nonzero. We have in $\cQ$
\begin{equation}\label{bez1}
A^{-1}B=(DU)^{-1}(1-DV)
\end{equation}
and similarly
\begin{equation}\label{bez2}
B^{-1}A=(CV)^{-1}(1-CU)^{-1}
\end{equation}
Since by assumption $A$ and $B$ are left coprime there exist two difference operators $P$ and $Q$ such that
\begin{equation}\label{bez3}
\begin{split}
&1-DV=PB, \hspace{2 mm} DU=PA\\
&1-CU=QA, \hspace{2 mm} CV=QB.
\end{split}
\end{equation}
Using the assumption $A(x)=B(y)$ and the first line of $\eqref{bez3}$ we get 
\begin{equation} \label{bez4}
y=(PB+DV)(y)=PA(x)+DV(y)=D(U(x)+V(y)),
\end{equation}
and similarly using the second line of \eqref{bez3} we get
\begin{equation} \label{bez5}
x=(CU+QA)(x)=CU(x)+QB(x)=C(U(x)+V(y)).
\end{equation}
Hence, the statement holds with $z=U(x)+V(y)$.
\end{proof}

\section*{Acknowledgements}
The paper is supported by AVM's EPSRC grant EP/P012655/1 and JPW's EPSRC grant 
EP/P012698/1. Both authors gratefully acknowledge the financial support.
JPW and SC were partially supported by Research in 
Pairs grant no. 41670 from the London Mathematical Society; SC also thanks the 
University 
of Kent for the hospitality received during his visit in July 2017. SC was supported by a Junior 
Fellow award from the Simons Foundation.
%\bibliography{kdv}

\end{document}